%% file: stack_tree.tex
\documentclass[a4paper]{llncs}

\usepackage{macro}
\newif\ifdraft\drafttrue

\ifdraft

\newcommand\am[1]{{\color{red} [#1 - \textbf{Antoine}]}}
\newcommand\vp[1]{{\color{blue} [#1 - \textbf{Vincent}]}}

\else

\newcommand\am[1]{}
\newcommand\vp[1]{}
\newcommand\review[1]{}

\fi

\begin{document}

\title{Rewriting Higher-Order Stack Trees \thanks{This work was
    partially supported by the French National Research Agency (ANR),
    through excellence program Bézout (ANR-10-LABX-58)}}

\author{Vincent Penelle}

\institute{Universit\'e Paris-Est, LIGM (CNRS UMR 8049), UPEM, CNRS,\\
 F-77454 Marne-la-Vall\'ee, France\\
\email{vincent.penelle@u-pem.fr}
}

\maketitle

\begin{abstract}
  Higher-order pushdown systems and ground tree rewriting systems can
  be seen as extensions of suffix word rewriting systems. Both classes
  generate infinite graphs with interesting logical
  properties. Indeed, the model-checking problem for
  monadic second order logic (respectively first order logic with
  a reachability predicate) is decidable on such graphs. We unify
  both models by introducing the notion of stack trees, trees whose nodes are
  labelled by higher-order stacks, and define the corresponding class
  of higher-order ground tree rewriting systems. We show that these
  graphs retain the decidability properties of ground tree rewriting
  graphs while generalising the pushdown hierarchy of graphs.
\end{abstract}

\section{Introduction}

Since Rabin's proof of the decidability of monadic second order logic (MSO) 
over 
the full infinite binary tree $\Delta_2$ \cite{Rabin68}, there has been an 
effort to characterise increasingly general classes of structures with 
decidable 
MSO theories. This can be achieved for instance using families of 
graph transformations which preserve the decidability of MSO - such as the 
unfolding or the MSO-interpretation and applying them 
to graphs of known decidable MSO theories, such as finite graphs or
the graph $\Delta_2$.

This approach was followed in \cite{Caucal96}, where it is shown 
that the prefix (or suffix) rewriting graphs of recognisable word rewriting 
systems, which coincide (up to graph isomorphism) with the transition graphs of 
pushdown automata (contracting $\varepsilon$-transitions), can be obtained from 
$\Delta_2$ using inverse regular substitutions, a simple class of 
MSO-compatible 
transformations. They also coincide with those obtained by applying MSO 
interpretations to $\Delta_2$ \cite{b98}.
Alternately unfolding and applying inverse regular mappings to these graphs 
yields a strict hierarchy of classes of trees and graphs with a decidable MSO 
theory \cite{Caucal02,CarayolW03} coinciding with the transition graphs of 
\emph{higher-order pushdown automata} and capturing the solutions of \emph{safe 
higher-order program schemes}\footnotemark, whose MSO decidability had already 
been established in \cite{knu02}. We will henceforth call this the 
\emph{pushdown hierarchy} and the graphs at its $n$-th level \emph{$n$-pushdown 
graphs} for simplicity.
\footnotetext{This hierarchy was extended to encompass \emph{unsafe}
  schemes and \emph{collapsible} automata, which are out of the scope
  of this paper. See \cite{bro10,car12,bro12} for recent results on
  the topic.}

Also well-known are the automatic and tree-automatic structures (see for 
instance \cite{bg04}), whose vertices are represented by words or trees and 
whose edges are characterised using finite automata running over tuples of 
vertices. The decidability of first-order logic (FO) over these graphs stems 
from the well-known closure properties of regular word and tree languages, but 
it can also be related to Rabin's result since tree-automatic graphs are 
precisely the class of graphs obtained from $\Delta_2$ using \emph{finite-set 
interpretations} \cite{ColcombetL07}, a generalisation of WMSO interpretations 
mapping structures with a decidable MSO theory to structures with a decidable 
FO 
theory.
Applying finite-set interpretations to the whole pushdown hierarchy therefore 
yields an infinite hierarchy of graphs of decidable FO theory, which is proven 
in \cite{ColcombetL07} to be strict.

Since prefix-recognisable graphs can be seen as word rewriting graphs, another 
variation is to consider similar rewriting systems over trees. This 
yields the class of \emph{ground tree rewriting graphs}, which strictly 
contains 
that of real-time order 1 pushdown graphs. This class is orthogonal to 
the whole pushdown hierarchy since it contains at least one graph of 
undecidable 
MSO theory, for instance the infinite 2-dimensional grid. The transitive 
closures of ground tree rewriting systems can be represented using \emph{ground 
tree transducers}, whose graphs were shown in \cite{DauchetT90} to have 
decidable FO[$\xrightarrow{*}$] theories by establishing their closure under 
iteration and then showing that any such graph is tree-automatic.

The purpose of this work is to propose a common extension to both higher-order 
stack operations and ground tree rewriting. We introduce a model of 
\emph{higher-order ground tree rewriting} over trees labelled by higher-order 
stacks (henceforth called \emph{stack trees}), which coincides, at
order 1, with ordinary ground tree rewriting and, over unary trees,
with the dynamics of  
higher-order pushdown automata. Following ideas from the works 
cited above, as well as the notion of recognisable sets and relations over 
higher-order stacks defined in \cite{Carayol05}, we introduce the class of 
\emph{ground (order $n$) stack tree rewriting systems}, whose derivation 
relations are captured by \emph{ground stack tree transducers}. Establishing 
that this class of relations is closed under iteration and can be finite-set 
interpreted in $n$-pushdown graphs yields the decidability of their 
FO[$\xrightarrow{*}$] theories.

The remainder of this paper is organised as follows. Section \ref{sec:def} 
recalls some of the concepts used in the paper. Section \ref{sec:host} 
defines stack trees and stack tree rewriting systems.
Section \ref{sec:gstt} explores a 
notion of recognisability for binary relations over stack trees. 
Section \ref{sec:rew_graph} proves the decidability of 
FO[$\xrightarrow{*}$] model checking over ground stack tree rewriting graphs.
Finally, Section \ref{sec:concl} presents some further perspectives.

\section{Definitions and notations}\label{sec:def}

\paragraph*{Trees.}
Given an arbitrary set $\Sigma$, an ordered $\Sigma$-labelled tree $t$ of 
arity at most $d \in \mathbb{N}$ 
is a \emph{partial function} from $\{1, \ldots, d\}^*$ to $\Sigma$ such
that the domain of $t$, $\dom(t)$ is prefix-closed (if $u$ is in
$\dom(t)$, then every prefix of $u$ is also in $\dom(t)$) and 
left-closed (for all $u \in \{1, \ldots, d\}^*$ and $2 \leq j \leq d$, $t(uj)$
is defined only if $t(ui)$ is for every $i < j$). Node $uj$ is called
the $j$-th \emph{child} of its \emph{parent} node $u$. Additionally, the
nodes of $t$ are totally ordered by the natural
length-lexicographic ordering $\leq_\llex$ over $\{1, \ldots, d\}^*$.
By abuse of notation, given a symbol $a\in \Sigma$, we simply denote
by $a$ the tree $\{\epsilon \mapsto a\}$ reduced to a unique
$a$-labelled node.
The frontier of $t$ is the set $\fr(t) = \{u \in \dom(t) \mid u1
\not\in \dom(t)\}$. Trees will always be drawn in such a way that the
left-to-right placement of leaves respects $\leq_\lex$.
The set of trees labelled by $\Sigma$ is denoted by $\mathcal{T}(\Sigma)$.
In this paper we only consider finite trees, i.e. trees with finite 
domains.

Given nodes $u$ and $v$, we write $u \sqsubseteq v$ if $u$ is a prefix of $v$, 
i.e. if there exists $w \in \{1,\cdots,d\}^*$, $v = uw$. We will say that 
$u$ is an \emph{ancestor} of $v$ or is \emph{above} $v$, and symmetrically that 
$v$ is \emph{below} $u$ or is its \emph{descendant}.
We call $v_{\leq i}$ the prefix of $v$ of length $i$.
For any $u \in \dom(t)$, $t(u)$ is called the \emph{label} of node $u$ in $t$ 
and $t_u = \{v \mapsto t(uv) \mid uv \in \dom(t)\}$ is the sub-tree of $t$
rooted at $u$.
For any $u \in \dom(t)$, we call $\#_t(u)$ the \emph{arity} of $u$, i.e. its 
number of children. When $t$ is understood, we simply write $\#(u)$.
Given trees $t, s_1, \ldots, s_k$ and a $k$-tuple of positions
$\mathbf{u} = (u_1, \ldots, u_k) \in \dom(t)^k$, we denote by
$\context{t}{\mathbf{u}}{s_1, \ldots s_k}$ the tree obtained by replacing the
sub-tree at each position $u_i$ in $t$ by $s_i$, i.e. the tree in which any node
$v$ not below any $u_i$ is labelled $t(v)$, and any node $u_i.v$ with
$v\in \dom(s_i)$ is labelled $s_i(v)$. In the special case
where $t$ is a $k$-\emph{context}, i.e. contains leaves $u_1, \ldots,
u_k$ labelled by special symbol $\diamond$, we omit $\mathbf{u}$ and
simply write $t[s_1, \ldots, s_k] = t[s_1, \ldots, s_k]_\mathbf{u}$.


\paragraph*{Directed Graphs.}
A \emph{directed graph} $G$ with edge labels in $\Gamma$ is a pair
$(V_G,E_G)$ where $V_G$ is a set of vertices and $E_G \subseteq (V_G
\times \Gamma \times V_G)$ is a set of edges. Given two vertices $x$
and $y$, we write $x \xrightarrow{\gamma}_G y$ if
$(x,\gamma,y)\in E_G$, $x \xrightarrow{}_G y$ if there
exists $\gamma\in \Gamma$ such that $x \xrightarrow{\gamma}_G y$,
and $x \xrightarrow{\Gamma'}_G y$ if there exists $\gamma\in
\Gamma'$ such that $x \xrightarrow{\gamma}_G y$.
There is a \emph{directed path} in $G$ from $x$ to $y$ labelled
by $w = w_1 \ldots w_k \in \Gamma^*$, written $x \xrightarrow{w}_G y$,
if there are vertices $x_0, \ldots, x_k$ such that $x = x_0$, $x_k =
y$ and for all $1 \leq i \leq k$, $x_{i-1} \xrightarrow{w_i}_G x_i$. We
additionally write $x \xrightarrow{*}_G y$ if there exists $w$ such that
$x \xrightarrow{w}_G y$ , and $x \xrightarrow{+}_G y$ if there is such a
path with $|w| \geq 1$.
A directed graph $G$ is \emph{connected} if there exists an
\emph{undirected} path between any two vertices $x$ and $y$, meaning
that $(x, y) \in (\xrightarrow{}_G \cup \xrightarrow{}_G^{-1})^*$.
%
We omit $G$ from all these notations when it is clear from the context.
%
%
A directed graph $D$ is \emph{acyclic}, or is a DAG, if there is no
$x$ such that $x \xrightarrow{+} x$. The \emph{empty
  DAG} consisting of a single vertex (and no edge, hence its name) is
denoted by $\emptydag$.
Given a DAG $D$, we denote by $I_D$ its set of vertices of in-degree
$0$, called \emph{input vertices}, and by $O_D$ its set of vertices of
out-degree $0$, called \emph{output vertices}. The DAG is said to be of
\emph{in-degree} $|I_D|$ and of \emph{out-degree} $|O_D|$.
We henceforth only consider finite DAGs.

\paragraph*{Rewriting Systems.}

Let $\Sigma$ and $\Gamma$ be finite alphabets. A $\Gamma$-labelled
\emph{ground tree rewriting system} (GTRS) is a finite set $R$ of
triples $(\ell, a, r)$ called \emph{rewrite rules}, with $\ell$ and
$r$ finite $\Sigma$-labelled trees and $a \in \Gamma$ a label. The
rewriting graph of $R$ is $\mathcal{G}_R = (V,E)$, where
$V = \mathcal{T}(\Sigma)$ and
$E = \{ (c[\ell], a, c[r]) \mid (\ell, a, r) \in R \}$. The
\emph{rewriting relation} associated to $R$ is
$\xrightarrow{}_R\ =\ \xrightarrow{}_{\mathcal{G}_R}$, its
\emph{derivation relation} is
$\xrightarrow{*}_R\ =\ \xrightarrow{*}_{\mathcal{G}_R}$.
When restricted to words (or equivalently unary trees), such systems
are usually called \emph{suffix} (or \emph{prefix}) \emph{word
  rewriting systems}.

\section{Higher-Order Stack Trees}
\label{sec:host}

\subsection{Higher-Order Stacks}

We briefly recall the notion of higher-order stacks (for details, see
for instance \cite{Carayol05}). In order to obtain a more
straightforward extension from stacks to stack trees, we use a
slightly tuned yet equivalent definition, whereby the hierarchy starts
at level $0$ and uses axs different set of basic operations.

In the remainder, $\Sigma$ will denote a fixed finite alphabet and $n$
a positive integer. We first define stacks of order $n$ (or
$n$-stacks).  Let $\Stacks_{0}(\Sigma) = \Sigma$ denote the set of
$0$-stacks.  For $n >0$, the set of $n$-stacks is $\Stacks_{n}(\Sigma)
= (\Stacks_{n-1}(\Sigma))^+$, the set of non-empty sequences of
$(n-1)$-stacks. When $\Sigma$ is understood, we simply write
$\Stacks_{n}$.  For $s\in \Stacks_{n}$, we write $s = \stack{n}{s_1,
  \cdots, s_k}$, with $k>0$ and $n>0$, for an $n$-stack of size $|s| = k$
whose topmost $(n-1)$-stack is $s_k$.
For example, $\stack{3}{\stack{2}{\stack{1}{aba}}
  \stack{2}{\stack{1}{aba} \stack{1}{b} \stack{1}{aa}}}$ is a
$3$-stack of size 2, whose topmost $2$-stack $\stack{2}{\stack{1}{aba}
  \stack{1}{b} \stack{1}{aa}}$ contains three $1$-stacks, etc.
%
%
\paragraph*{Basic Stack Operations.}
Given two letters $a,b \in \Sigma$, we define the partial function $\rew{a}{b} 
: \Stacks_{0} \rightarrow \Stacks_{0}$ such that $\rew{a}{b}(c) = 
b$, if $c=a$ and is not defined otherwise. We also consider the identity 
function $\Id : \Stacks_{0} \rightarrow \Stacks_{0}$. For $n \geq 1$, the 
function $\cop{n} : \Stacks_{n} \rightarrow \Stacks_{n}$ is 
defined by $\cop{n}(s) = \stack{n}{s_1,\cdots,s_k,s_k}$, for every 
$s=\stack{n}{s_1,\cdots,s_k} \in \Stacks_{n}$. As it is injective, we 
denote by $\ncop{n}$ its inverse (which is a partial function).

Each level $\ell$ operation $\theta$ is extended to any level $n > \ell$
stack $s = \stack{n}{s_1,\cdots,s_k}$ by letting $\theta(s) =
\stack{n}{s_1,\cdots,s_{k-1},\theta(s_k)}$.  The set $\Ops{n}$ of
basic operations of level $n$ is defined as: $\Ops{0} = \{\rew{a}{b}
\mid a,b \in \Sigma\} \cup \{\Id\}$, and for $n\geq 1$, $\Ops{n} =
\Ops{n-1} \cup \{\cop{n},\ncop{n}\}$.



\subsection{Stack Trees}

We introduce the set
$\TS_{n}(\Sigma) = \mathcal{T}(\Stacks_{n-1}(\Sigma))$ (or simply
$\TS_n$ when $\Sigma$ is understood) of \emph{$n$-stack-trees}.
Observe that an $n$-stack-tree of
degree 1 is isomorphic to an $n$-stack, and that $\TS_{1} = 
\mathcal{T}(\Sigma)$.
Figure \ref{fig:stack_tree} shows an example of a 3-stack tree.
The notion of stack trees therefore subsumes both higher-order stacks and
ordinary trees.

\begin{figure}[t]
  \begin{center}
    \begin{tikzpicture}[scale=.7]
      \node (root) at (0,0) {$\stack{2}{\stack{1}{aa}\stack{1}{bab}}$}; 
      \node (0) at (-2,-1.25) {$\stack{2}{\stack{1}{aa}\stack{1}{aaa}}$}; 
      \node (1) at (3,-1.25) {$\stack{2}{\stack{1}{aa}\stack{1}{a}\stack{1}{b}}$};
      \node (00) at (-3,-2.5) {$\stack{2}{\stack{1}{ab}}$}; 
      \node (10) at (1.5,-2.5) {$\stack{2}{\stack{1}{ba}\stack{1}{ba}\stack{1}{b}}$}; 
      \node (11) at (5,-2.5) {$\stack{2}{\stack{1}{abb}\stack{1}{ab}}$};
      \draw [->] (root) to (0); 
      \draw [->] (root) to (1); 
      \draw [->] (0) to (00); 
      \draw [->] (1) to (10); 
      \draw [->] (1) to (11);
    \end{tikzpicture}
  \end{center}
  \caption{A 3-stack-tree.}\label{fig:stack_tree}
\end{figure}
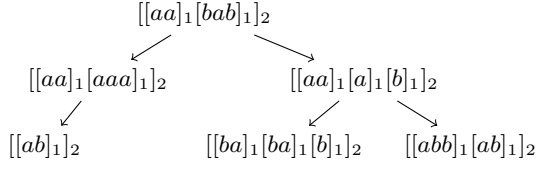

\paragraph*{Basic Stack Tree Operations.}
We now extend $n$-stack operations to stack trees. There are in
general several positions where one may perform a given operation on a
tree. We thus first define the \emph{localised} application of an
operation to a specific position in the tree (given by the index of a
leaf in the lexicographic ordering of leaves), and then derive a
definition of stack tree operations as binary relations, or
equivalently as partial functions from stack trees to sets of stack
trees.

Any operation of $\Ops{n-1}$ is extended to $\TS_{n}$ as follows: given
$\theta \in \Ops{n-1}$, and an integer $i\leq |\fr(t)|$, $\theta_{(i)}(t) = 
\context{t}{u_i}{\theta(s)}$ with $s=t(u_i)$, where $u_i$ is the $i^{th}$ leaf 
of the tree, with respect to the lexicographic order. If $\theta$ is not 
applicable to $s$, $\theta_i(t)$ is not defined.
We define $\theta(t) = \{\theta_{(i)}(t)\mid i\leq |\fr(t)|\}$, 
i.e. the set of stack trees obtained by applying $\theta$ to a leaf of 
$t$. 

The $k$-fold duplication of a stack tree leaf and its label is denoted by 
$\bcopy{n}{k} : \TS_{n} \rightarrow 2^{\TS_{n}}$. Its application to the 
$i^{th}$ 
leaf of a tree $t$ is: ${\bcopy{n}{k}}_{(i)}(t) = t\cup \{u_i j \mapsto t(u_i) 
\mid j \leq k\}$, with $i\leq |\fr(t)|$. Let $\bcopy{n}{k}(t)= 
\{{\bcopy{n}{k}}_{(i)}(t)\}$ be the set of stack trees obtained by applying 
$\bcopy{n}{k}$ to a leaf of $t$.
The inverse operation, written $\nbcopy{n}{k}$, is such that $t' =
{\nbcopy{n}{k}}_{(i)}(t)$ if $t = {\bcopy{n}{k}}_{(i)}(t')$. We also define 
$\nbcopy{n}{k}(t) = \{{\nbcopy{n}{k}}_{(i)}(t)\}$. Notice that $t' \in 
\nbcopy{n}{k}(t)$ if $t \in \bcopy{n}{k}(t')$.

For simplicity, we will henceforth only consider the case where stack trees 
have arity at most $2$ and $k \leq 2$, but all results go through in the 
general case. We denote by $\TOps{n} = \Ops{n-1} \cup 
\{\bcopy{n}{k},\nbcopy{n}{k} \mid k \leq 2\}$ the set of basic operations over 
$\TS_{n}$.

\subsection{Stack Tree Rewriting}


As already mentioned, $\TS_{1}$ is the set of trees labelled by
$\Sigma$. In contrast with basic stack tree operations, a tree rewrite
rule $(\ell, r)$ expresses the replacement of an arbitrarily large
ground subtree $\ell$ of some tree $s = c[\ell]$ into $r$, yielding
the tree $c[r]$. 
Contrary to the case of order 1 stacks (which are simply words),
composing basic stack tree operations does not allow us to directly
express such an operation, because there is no guarantee that two
successive operations will be applied to the same part of a tree.
We thus need to find a way to consider compositions of basic
operations acting on a single sub-tree.
In our notations, the effect of a ground tree rewrite rule could thus be seen as the
\emph{localised} application of a sequence of $\mathrm{rew}$ and
$\nbcopy{1}{2}$ operations followed by a sequence of $\mathrm{rew}$
and $\bcopy{1}{2}$ operations. The relative positions where these
operations must be applied could be represented as a pair of trees with edge
labels in $\Ops{0}$.

From level 2 on, this is no longer possible. Indeed a localised sequence of
operations may be used to perform introspection on the stack labelling
a node without destroying it, by first performing a $\cop{2}$
operation followed by a sequence of level 1 operations and a
$\ncop{2}$ operation. It is thus impossible to directly represent such
a transformation using pairs of trees labelled by stack tree
operations.
We therefore adopt a presentation of \emph{compound operations} as
DAGs, which allows us to specify the relative application positions of
successive basic operations. However, not every DAG represents a valid 
compound operation, so we first need to define a suitable subclass of
DAGs and associated concatenation operation.
An example of the model we aim to define can be found in Fig. 
\ref{fig:application_of_a_dag}.
\begin{figure}[t]
\subfloat[Stack tree $t$]{
     \begin{tikzpicture}[scale=.7]
        \node (root) at (0,.5) {$\stack{1}{bbb}$}; 
        \node (0) at (-1,-1) {$\stack{1}{bbb}$}; 
        \node (1) at (1,-1) {$\stack{1}{aabb}$}; 
        \draw [->] (root) to (0); 
        \draw [->] (root) to (1); 
      \end{tikzpicture}
}\hfil
\subfloat[Operation $D$]{
\begin{tikzpicture}[scale=.6]  
      \node (5p) at (0,-5) {.};
      \node (6) at (0,-6) {.};
      \node (7) at (0,-7) {.};  
      \node (8) at (-1,-8) {.};  
      \node (9) at (1,-8) {.};  
      \node (10) at (-1,-9) {.};  
      \node (11) at (1,-9) {.};  
      
      \draw[->] (5p) to node[midway,right]{$\ncop{1}$} (6) ;
      \draw[->] (6) to node[midway,right]{$\rew{b}{c}$} (7) ;
      \draw[->] (7) to node[near start,left]{$1$} (8) ;
      \draw[->] (7) to node[near start,right]{$2$} (9) ;
      \draw[->] (8) to node[midway,right]{$\rew{c}{a}$} (10) ;
      \draw[->] (9) to node[midway,right]{$\cop{1}$} (11) ;
\end{tikzpicture}
}\hfil
\subfloat[$D_{(1)}(t)$]{
     \begin{tikzpicture}[scale=.7]
	\node (root) at (0,.5) {$\stack{1}{bbb}$};
        \node (0) at (-1,-1) {$\stack{1}{bc}$}; 
        \node (1) at (1,-1) {$\stack{1}{aabb}$};
        \node (00) at (-2,-2.5) {$\stack{1}{ba}$}; 
        \node (01) at (0,-2.5) {$\stack{1}{bcc}$}; 
        \draw [->] (root) to (0); 
        \draw [->] (root) to (1); 
        \draw [->] (0) to (00); 
        \draw [->] (0) to (01);
      \end{tikzpicture}
}\hspace{.05cm}
\subfloat[$D_{(2)}(t)$]{
     \begin{tikzpicture}[scale=.7]
	\node (root) at (0,.5) {$\stack{1}{bbb}$};
        \node (0) at (-1,-1) {$\stack{1}{bbb}$}; 
        \node (1) at (1,-1) {$\stack{1}{aac}$};
        \node (10) at (0,-2.5) {$\stack{1}{aaa}$}; 
        \node (11) at (2,-2.5) {$\stack{1}{aacc}$}; 
        \draw [->] (root) to (0); 
        \draw [->] (root) to (1); 
        \draw [->] (1) to (10); 
        \draw [->] (1) to (11); 
      \end{tikzpicture}
}
\caption{The application of an operation $D$ to a stack tree $t$.
}\label{fig:application_of_a_dag}
\end{figure}
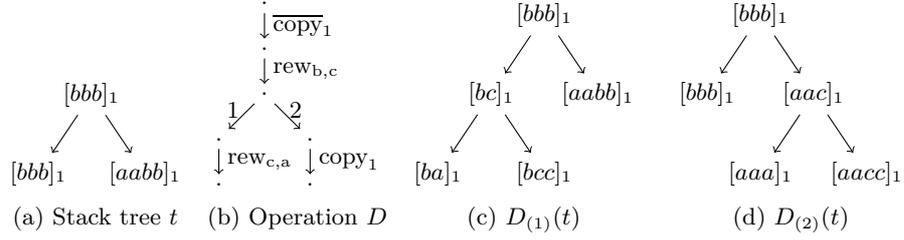

\paragraph*{Concatenation of DAGs.}

Given two DAGs $D$ and $D'$ with $O_D = \{b_1, \ldots, b_\ell\}$ and
$I_{D'} = \{a'_1, \ldots, a'_{k'}\}$ and two indices $i$ and $j$ with
$1 \leq i \leq \ell$ and $1 \leq j \leq k'$, we denote by $D
\cdot_{i,j} D'$ the unique DAG $D''$ obtained by merging the
$(i+m)$-th output vertex of $D$ with the $(j+m)$-th input vertex of
$D'$ for all $m \geq 0$ such that both $b_{i+m}$ and $a'_{j+m}$
exist. 
Formally, letting $d = \min(\ell-i, k'-j)+1$ denote the number of
merged vertices, we have $D'' = \mathrm{merge}_f (D \uplus D')$ where
$\mathrm{merge}_f(D)$ is the DAG whose set of vertices is $f(V_D)$ and
set of edges is $\{(f(x),\gamma,f(x'))\mid (x,\gamma,x')\in E_D\}$,
and $f(x) = b_{i+m}$ if $x = a'_{j+m}$ for some $0 \leq m \leq d$, and
$f(x) = x$ otherwise. We call $D''$ the $(i,j)$-concatenation of $D$
and $D'$. Note that the $(i,j)$-concatenation of two connected DAGs
remains connected.

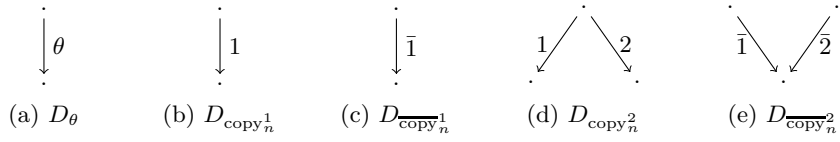
\begin{figure}[t]
\subfloat[$D_\theta$]{
\begin{tikzpicture}
  \useasboundingbox (-.7,-.5) rectangle (.7,.7);
      \node (1) at (0,.6) {.};
      \node (2) at (0,-.4) {.}; 
      \draw[->] (1) to node[midway,right]{$\theta$} (2) ; 
\end{tikzpicture}
}\hfil
\subfloat[$D_{\bcopy{n}{1}}$]{
\begin{tikzpicture}
  \useasboundingbox (-.7,-.5) rectangle (.7,.7);
      \node (1) at (0,.6) {.};
      \node (2) at (0,-.4) {.}; 
      \draw[->] (1) to node[midway,right]{$1$} (2) ; 
\end{tikzpicture}
}\hfil
\subfloat[$D_{\nbcopy{n}{1}}$]{
\begin{tikzpicture}
  \useasboundingbox (-.7,-.5) rectangle (.7,.7);
      \node (1) at (0,.6) {.};
      \node (2) at (0,-.4) {.}; 
      \draw[->] (1) to node[midway,right]{$\bar{1}$} (2) ; 
\end{tikzpicture}
}\hfil
\subfloat[$D_{\bcopy{n}{2}}$]{
\begin{tikzpicture}
      \node (1) at (0,.5) {.}; 
      \node (2) at (-.7,-.5) {.}; 
      \node (3) at (.7,-.5) {.}; 
      \draw [->] (1) to node[midway,left]{$1$} (2); 
      \draw [->] (1) to node[midway,right]{$2$} (3) ; 
\end{tikzpicture}
}\hfil
\subfloat[$D_{\nbcopy{n}{2}}$]{
\begin{tikzpicture}
      \node (2) at (-.7,.5) {.};
      \node (3) at (.7,.5) {.};
      \node (4) at (0,-.5) {.};
      \draw [->] (2) to node[midway,left] {$\bar{1}$} (4);
      \draw [->] (3) to node[midway,right]{$\bar{2}$} (4);
\end{tikzpicture}
}
\caption{DAGs of the basic $n$-stack tree operations (here $\theta$ 
ranges over $\Ops{n-1}$).}
\label{inj_tops_dag}
\end{figure}

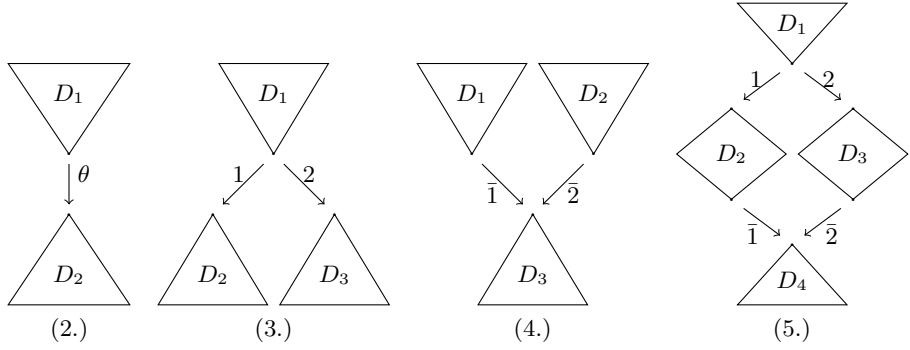
\begin{figure}[t]
\stepcounter{subfigure}
\renewcommand{\thesubfigure}{\arabic{subfigure}.}
\subfloat[]{
\begin{tikzpicture}[scale=.8]
      \node (0) at (0,1) {$D_1$};
      \node (1) at (0,0) {.};
      \node (2) at (0,-1) {.}; 
      \node (4) at (0,-2) {$D_2$};
      \draw[->] (1) to node[near start,right]{$\theta$} (2) ; 
      \draw (0,-1) -- (1,-2.5) -- (-1,-2.5) -- cycle;
      \draw (0,0) -- (1,1.5) -- (-1,1.5) -- cycle;
\end{tikzpicture}
}\hfil
\subfloat[]{
\begin{tikzpicture}[scale=.8]
      \node (1) at (0,0) {.};
      \node (2) at (-1,-1) {.};
      \node (2') at (1,-1) {.}; 
      \node (4) at (-1,-2) {$D_{2}$};
      \node (4') at (1,-2) {$D_{3}$};
      \node at (0,1) {$D_1$};
      \draw[->] (1) to node[near start,left]{$1$} (2) ;
      \draw[->] (1) to node[near start,right]{$2$} (2') ; 
      \draw (-1,-1) -- (-.1,-2.5) -- (-1.9,-2.5) -- cycle; 
      \draw (1,-1) -- (1.9,-2.5) -- (0.1,-2.5) -- cycle;
      \draw (0,0) -- (.9,1.5) -- (-.9,1.5) -- cycle;
\end{tikzpicture}
}\hfil
\subfloat[]{
\begin{tikzpicture}[scale=.8]
      \node (1) at (-1,0) {.};
      \node (1') at (1,0) {.};
      \node (2) at (0,-1) {.}; 
      \node (4) at (0,-2) {$D_3$};
      \node at (-1,1) {$D_1$};
      \node at (1,1) {$D_2$};
      \draw[->] (1) to node[near start,below]{$\bar{1}$} (2) ;
      \draw[->] (1') to node[near start,below]{$\bar{2}$} (2) ; 
      \draw (0,-1) -- (.9,-2.5) -- (-.9,-2.5) -- cycle; 
      \draw (-1,0) -- (-.1,1.5) -- (-1.9,1.5) -- cycle; 
      \draw (1,0) -- (.1,1.5) -- (1.9,1.5) -- cycle;
\end{tikzpicture}
}\hfil
\subfloat[]{
\begin{tikzpicture}[scale=.8]
      \node (1) at (-1,-.75) {.};
      \node (1') at (1,-.75) {.};
      \node (2) at (0,-1.5) {.}; 
      \node (4) at (-1,.75) {.};
      \node (4') at (1,.75) {.};
      \node (3) at (0,1.5) {.}; 
      \node at (0,-2.1) {$D_4$};
      \node at (-1,0) {$D_2$};
      \node at (1,0) {$D_3$};
      \node at (0,2.15) {$D_1$};
      \draw[->] (1) to node[near start,below]{$\bar{1}$} (2) ;
      \draw[->] (1') to node[near start,below]{$\bar{2}$} (2) ;
      \draw[->] (3) to node[near start,left]{$1$} (4) ;
      \draw[->] (3) to node[near start,right]{$2$} (4') ;
      \draw (0,-1.5) -- (.9,-2.5) -- (-.9,-2.5) -- cycle; 
      \draw (-1,-.75) -- (-.1,0) -- (-1,.75) -- (-1.9,0) -- cycle; 
      \draw (1,-.75) -- (.1,0) -- (1,.75) -- (1.9,0) -- cycle;
      \draw (0,1.5) -- (.9,2.5) -- (-.9,2.5) -- cycle;
\end{tikzpicture}
}
\caption{
  Possible decompositions of a compound operation, numbered according
  to the items in Definition \ref{def:compound}.  }
\label{fig_application}
\end{figure}

\paragraph*{Compound Operations}
We represent compound operations as DAGs.  We will refer in particular
to the set of DAGs
$\graphops{n} = \{D_\theta \mid \theta \in \TOps{n}\}$ associated with
basic operations, which are depicted in Fig.  \ref{inj_tops_dag}.
Compound operations are inductively defined below, as depicted in
Fig. \ref{fig_application}.

\begin{definition}
  \label{def:compound}
  A DAG $D$ is a \emph{compound operation} (or simply an
  \emph{operation}) if one of the following holds:
 \begin{enumerate}
  \item $D = \emptydag$;
  \item $D = (D_1 \cdot_{1,1} D_\theta) \cdot_{1,1} D_2$, with
    $|O_{D_1}| = |I_{D_2}| = 1$ and
    $\theta \in \Ops{n-1} \cup \{\bcopy{n}{1}, \nbcopy{n}{1}\}$;
  \item
    $D = ((D_1 \cdot_{1,1} D_{\bcopy{n}{2}}) \cdot_{2,1} D_3)
    \cdot_{1,1} D_2$, with $|O_{D_1}| = |I_{D_2}| = |I_{D_3}| = 1$;
  \item
    $D = (D_1 \cdot_{1,1} (D_2 \cdot_{1,2} D_{\nbcopy{n}{2}}))
    \cdot_{1,1} D_3$ with $|O_{D_1}| = |O_{D_2}| = |I_{D_3}| = 1$;
  \item \label{item:diamond}
    $D = ((((D_1 \cdot_{1,1} D_{\bcopy{n}{2}}) \cdot_{2,1} D_3)
    \cdot_{1,1} D_2) \cdot_{1,1} D_{\nbcopy{n}{2}}) \cdot_{1,1} D_4$,
    with
    $|O_{D_1}| = |I_{D_2}| = |O_{D_2}| = |I_{D_3}| = |O_{D_3}| =
    |I_{D_4}| = 1$ ;
 \end{enumerate}
 where $D_1, D_2, D_3$ and $D_4$ are compound operations.
\end{definition}

Additionally, the vertices of $D$ are ordered inductively in such a
way that every vertex of $D_i$ in the above definition is smaller than
the vertices of $D_{i+1}$, the order over $\emptydag$ being the empty
one. This induces in particular an order over the input vertices of
$D$, and one over its output vertices.


\begin{definition}
  Given a compound operation $D$, we define $D_{(i)}(t)$, its
  \emph{localised application} starting at the $i$-th leaf of a stack
  tree $t$, as follows:
\begin{enumerate}
\item If $D = \emptydag$, then $D_{(i)} (t) = t$.
\item If $D = (D_1 \cdot_{1,1} D_\theta) \cdot_{1,1} D_2$ with
  $\theta \in \Ops{n-1} \cup \{\bcopy{n}{1},\nbcopy{n}{1}\}$,
  
   \hfill then $D_{(i)}(t) = {D_2}_{(i)}(\theta_{(i)}({D_1}_{(i)}(t)))$.

\item If $D = ((D_1 \cdot_{1,1} D_{\bcopy{n}{2}}) \cdot_{2,1} D_3)
  \cdot_{1,1} D_2$,
  
  \hfill then $D_{(i)}(t) =
  {D_2}_{(i)}({D_3}_{(i+1)}({\bcopy{n}{2}}_{(i)}({D_1}_{(i)}(t))))$.

\item If $D = ((D_1 \cdot_{1,1} (D_2 \cdot_{2,1} D_{\nbcopy{n}{2}})) 
  \cdot_{1,1} D_3$,

  \hfill then $D_{(i)}(t) = {D_3}_{(i)}({\nbcopy{n}{2}}_{(i)}($
  ${D_2}_{(i+1)}({D_1}_{(i)}(t))))$.

\item If $D = ((((D_1 \cdot_{1,1} D_{\bcopy{n}{2}}) \cdot_{2,1} D_3)
  \cdot_{1,1} D_2) \cdot_{1,1} D_{\nbcopy{n}{2}}) \cdot_{1,1} D_4$,
  
  \hfill then $D_{(i)}(t) = {D_4}_{(i)}
  ({\nbcopy{n}{2}}_{(i)}({D_3}_{(i+1)}({D_2}_{(i)}({\bcopy{n}{2}}_{(i)}(
  {D_1}_{(i)}(t))))))$.
\end{enumerate}

\end{definition}

\begin{remark}
  An operation may admit several different decompositions with respect
  to Def. \ref{def:compound}. However, its application is
  well-defined, as one can show this process is locally confluent.
\end{remark}

Given two stack trees $t$, $t'$ and an operation $D$, we say that
$t' \in D(t)$ if there is a position $i$ such that $t' =
D_{(i)}(t)$. Figure \ref{fig:application_of_a_dag} shows an
example. We call $\mathcal{R}_D$ the relation induced by
$D$: for any stack trees $t,t'$, $\mathcal{R}_D(t,t')$ if
and only if $t'\in D(t)$.
Finally, given a $k$-tuple of operations $\bar{D} = (D_1,
\ldots, D_k)$ of respective in-degrees $d_1,
\ldots, d_k$ and a $k$-tuple of indices $\mathbf{i} = (i_1, \ldots,
i_k)$ with $i_{j+1} \geq i_{j} + d_j$ for all $1 \leq j < k$, we denote
by $\bar{D}_{(\mathbf{i})}(t)$ the parallel application
${D_1}_{(i_1)} (\ldots {D_k}_{(i_k)} (t) \ldots )$ of
$D_1, \ldots, D_k$ to $t$, $\bar{D}(t)$ the set of
all such applications and $\mathcal{R}_{\bar{D}}$ the induced relation.


Since the $(i,j)$-concatenation of two operations as defined above is not
necessarily a licit operation, we need to restrict ourselves to
results which are well-formed according to Def. \ref{def:compound}.
Given $D$ and $D'$, we let $D \cdot D' = \{D \cdot_{i,j} 
D' \mid D \cdot_{i,j} D' \text{ is an operation}\}$. Given $n >1$, we 
define\footnotemark $D^n = \bigcup_{i <n} D^i \cdot D^{n-i}$, and let $D^* 
= \bigcup_{n\geq 0} D^n$ denote the set of \emph{iterations} of $D$.
These notations are naturally extended to sets of operations.

\footnotetext{This unusual definition is necessary because $\cdot$ is not 
associative. For example, $(D_{\bcopy{n}{2}} \cdot_{2,1}
D_{\bcopy{n}{2}}) \cdot_{1,1}  D_{\bcopy{n}{2}}$ is in
$(D_{\bcopy{n}{2}})^2 \cdot D_{\bcopy{n}{2}}$ but not  
in $D_{\bcopy{n}{2}} \cdot (D_{\bcopy{n}{2}})^2$. }

\begin{proposition}
  $\graphops{n}^*$ is precisely the set of all well-formed compound
  operations.
\end{proposition}

\begin{proof}
  Recall that $\graphops{n}$ denotes the set of DAGs associated with
  basic operations. By definition of iteration, any DAG in
  $\graphops{n}^*$
  is an operation.
  Conversely, by Def. \ref{def:compound}, any operation can be
  decomposed into a concatenation of DAGs of $\graphops{n}$.  \qed
\end{proof}

\paragraph*{Ground Stack Tree Rewriting Systems.}
By analogy with order 1 trees, given some finite alphabet of labels
$\Gamma$, we call any finite subset of labelled operations in
$\graphops{n}^* \times \Gamma$ a labelled \emph{ground stack-tree
  rewriting system} (GSTRS). We straightforwardly extend the notions
of rewriting graph and derivation relation to these systems.
Note that for $n = 1$, this class coincides with ordinary ground tree
rewriting systems.  Moreover, one can easily show that the rewriting
graphs of ground stack-tree rewriting systems over unary $n$-stack
trees (trees containing only unary operations, i.e. no edge labelled
by $2$ or $\bar{2}$) are isomorphic to the configuration graphs of
order $n$ pushdown automata performing a finite sequence of operations
at each transition.

\section{Operation Automata}\label{sec:gstt}

In this section, in order to provide finite descriptions of possibly
infinite sets of operations, in particular the derivation relations of
GSTRS, we extend the notion of \emph{ground tree transducers} (or GTT)
of \cite{DauchetT90} to ground tree rewriting systems. 

A GTT $T$ is given by a tuple $\big((A_i, B_i)\big)_{1 \leq i \leq k}$
of pairs of finite tree automata. A pair of trees $(s,t)$ is accepted
by $T$ if $s = c[s_1, \ldots s_m]$ and $t = c[t_1, \ldots, t_m]$ for
some $m$-context $c$, where for all $1\leq j \leq m$, $s_j \in L(A_i)$
and $t_j \in L(B_i)$ for some $1 \leq i \leq k$. It is also shown
that, given a relation $R$ recognised by a GTT, there exists another
GTT recognising its reflexive and transitive closure $R^*$.

Directly extending this idea to ground stack tree rewriting systems is
not straightforward: 
contrary to the case of trees, a given compound operation may be applicable to 
many different subtrees. Indeed, the only subtree to which a ground tree 
rewriting rule $(s,t)$ 
can be applied is the tree $s$. On stack trees, this is no longer true, as 
depicted in Fig. \ref{fig:application_of_a_dag}: an operation does not entirely
describe the labels of nodes of subtrees it can be applied to (as in 
the case of trees), and can therefore be applied to infinitely many different 
subtrees.
We will thus express relations by
describing sets of compound operations over stack trees. 
Following \cite{Carayol05} where recognisable sets of higher-order stacks are
defined, we introduce operation automata and recognisable sets of operations. 

\begin{definition}
An automaton over $\graphops{n}^*$ is a tuple $A = 
(Q,\Sigma,I,F,\Delta)$, where
\begin{itemize}
\item $Q$ is a finite set of states,
\item $\Sigma$ is a finite stack alphabet,
\item $I \subseteq Q$ is a set of initial states,
\item $F \subseteq Q$ is a set of final states,
\item $\Delta \subseteq \left( Q\times
  (\Ops{n-1}\cup\{\bcopy{n}{1},\nbcopy{n}{1}\}) \times Q \right)$

\hfill $\cup \left( (Q\times Q) \times Q \right) \cup \left(Q \times (Q
  \times Q) \right)$ is a set of transitions.
\end{itemize}
\end{definition}

An operation $D$ is accepted by $A$ if there is a labelling of its 
vertices by states of $Q$ such that all input vertices are labelled by 
initial states, all output vertices by final states, and this labelling
is consistent with $\Delta$, in the sense that for  
all $x$, $y$ and $z$ respectively labelled by states $p$, $q$ and $r$, and for 
all $\theta \in \Ops{n-1} \cup \{\bcopy{n}{1},\nbcopy{n}{1}\}$,
\begin{align*}
  x \xrightarrow{\theta} y  & \implies (p,\theta,q) \in \Delta,
  \\
  x \xrightarrow{1} y \land x \xrightarrow{2} z & 
\implies (p,(q,r)) \in \Delta,
  \\
  x \xrightarrow{\bar{1}} z \land y \xrightarrow{\bar{2}} z & 
\implies ((p,q),r) \in \Delta.
\end{align*}

We denote by 
$\mathrm{Op}(A)$ the set of operations recognised by $A$. $\Rec$ 
denotes the class of sets of operations recognised by operation automata.
A pair of stack trees $(t,t')$ is in the relation $\mathcal{R}(A)$
defined by $A$ if for some $k \geq 1$ there is a $k$-tuple of
operations $\bar{D} = (D_1,\cdots,D_k)$ in $\graphaut{A}^k$
such that $t' \in \bar{D}(t)$.
At order $1$, we have already seen that stack trees are simply trees,
and that ground stack tree rewriting systems coincide with ground tree
rewriting systems. Similarly, we also have the following:

\begin{proposition}
  The classes of relations recognised by order $1$ operation
  automataand by ground tree transducers coincide.
\end{proposition}

At higher orders, the class $\Rec$ and the corresponding binary
relations retains several of the good closure properties of ground
tree transductions.

\begin{proposition}\label{prn:closure}
  $\Rec$ is closed under union, intersection and iterated
  concatenation. The class of relations defined by operation automata
  is closed under composition and iterated composition.
\end{proposition}
 
The construction of automata recognising the union and intersection of
two recognisable sets, the iterated concatenation of a recognisable
set, or the composition of two automata-definable relations, can be
found in the appendix. Given automaton $A$, the relation defined by
the automaton accepting $\mathrm{Op}(A)^*$ is $\mathcal{R}(A)^*$.

\paragraph*{Normalised automata.}
Operations may perform ``unnecessary'' actions on a given
stack tree, for instance duplicating a leaf with a $\bcopy{n}{2}$
operation and later destroying both copies with
$\nbcopy{n}{2}$. Such 
operations 
which leave the input tree unchanged are referred to as \emph{loops}. 
There are thus in general infinitely many operations
representing the same relation over stack trees. It is therefore
desirable to look for a canonical representative (a canonical
operation) for each considered relation. The intuitive idea is to
simplify operations by removing occurrences of successive mutually
inverse basic operations. This process is a very classical tool in the
literature of pushdown automata and related models, and was applied to
higher-order stacks in \cite{Carayol05}.  Our notion of reduced
operations is an adaptation of this work.

There are two main hurdles to overcome. First, as already mentioned, a
compound
operation $D$ can perform introspection on the label of a leaf
without destroying it. If $D$ can be applied to a given stack
tree $t$, such a sequence of operations does not change the resulting
stack tree $s$. It does however forbid the application of $D$ to
other stack trees by inspecting their node labels, hence 
removing this part of the computation would lead to an operation
with a possibly strictly larger domain. To adress this problem, and
following \cite{Carayol05}, we use \emph{test
  operations} ranging over regular sets of $(n-1)$-stacks, which will
allow us to handle non-destructive node-label introspection.

A second difficulty appears when an operation destroys a subtree
and then reconstructs it identically, for instance a $\nbcopy{n}{2}$
operation followed by $\bcopy{n}{2}$. Trying to remove such a pattern
would lead to a disconnected DAG, which does not describe a compound
operation in our sense. We thus need to leave such occurrences
intact. We can nevertheless bound the number of times a given position
of the input stack tree is affected by the application of an operation by
considering two phases: a \emph{destructive} phase during which only
$\nbcopy{n}{i}$ and order $n-1$ basic operations (possibly including tests)
are performed on the input stack-tree, and a \emph{constructive} phase
only consisting of $\bcopy{n}{i}$ and order $n-1$ basic operations.
Similarly to the way ground tree rewriting is
performed at order 1.

Formally, a \emph{test} $T_L$ over $\Stacks_{n}$ is the restriction of
the identity operation to $L \in \Rec(\Stacks_{n})$\footnote{Regular
  sets of $n$-stacks are
  obtained by considering regular sets of sequences of operations of
  $\Ops{n}$ applied to a given stack $s_0$. More details can be found
  in \cite{Carayol05}.}. In other words, given $s \in \Stacks_{n}$,
$T_L(s) = s$ if $s \in L$, otherwise, it is undefined. We denote 
by $\Tests{n}$ the set of
test operations over $\Stacks_{n}$. We enrich our basic operations
over $\TS_{n}$ with $\Tests{n-1}$. We
also extend compound operations with edges labelled by
tests. We denote by $\graphops{n}^\mathcal{T}$ the set of
basic operations with tests. We can now define the notion of
reduced operation analogously to that of reduced instructions with 
tests in \cite{Carayol05}.


\begin{definition}
  For $i \in \{0,\cdots,n\}$, we define the set of words
  $\mathrm{Red}_i$ over
  $\Ops{n}\cup\Tests{n}\cup\{1,2,\bar{1},\bar{2}\}$ as:
 \begin{itemize}
  \item $\mathrm{Red}_0 = \{\varepsilon,T,\rew{a}{b},\rew{a}{b} \cdot T,T \cdot
    \rew{a}{b},\rew{a}{b} \cdot T \cdot \rew{c}{d}$

    \hfill $\mid a,b,c,d \in \Sigma,T \in \Tests{n}\}$,
  \item For $0<i<n$, $\mathrm{Red}_i = (\mathrm{Red}_{i-1} \cdot \ncop{i})^* 
\cdot \mathrm{Red}_{i-1} \cdot (\cop{i}\cdot\mathrm{Red}_{i-1})^*$,
  \item $\mathrm{Red}_n = (\mathrm{Red}_{n-1}\cdot \{\bar{1},\bar{2}\} )^* \cdot
\mathrm{Red}_{n-1}\cdot (\{1,2\}\cdot \mathrm{Red}_{n-1})^*$.
 \end{itemize}
\end{definition}

%
%

\begin{definition}
 An operation with tests $D$ is \emph{reduced} if for every $x,y \in 
V_D$, if $x\xrightarrow{w} y$, then $w \in \mathrm{Red}_n$.
\end{definition}

Observe that, in the decomposition of a reduced operation $D$, case
\ref{item:diamond} of the inductive definition of compound operations
(Def. \ref{def:compound}) should never occur, as otherwise, there
would be a path on which $1$ appears before $\bar{1}$, which
contradicts the definition of reduced operation.

An automaton $A$ is said to be \emph{normalised} if it only accepts
reduced operations, and \emph{distinguished} if 
there is no transition ending in an initial state or starting in a final state. 
The following proposition shows that any
operation automaton can be normalised and distinguished.

\begin{proposition}
  For every automaton $A$, there exists a distinguished
  normalised automaton with tests $A_r$ such that
  $\mathcal{R}(A) = \mathcal{R}(A_r)$.
\end{proposition}

The idea of the construction is to 
transform $A$ in several steps, each modifying the set of accepted
operations but not the recognised relation. The proof relies on the
closure properties of regular sets of $(n-1)$-stacks and an analysis
of the structure of $A$.
We show in particular, using a saturation technique, that the set of
states of $A$ can be partitioned into \emph{destructive states} (which
label the destructive phase of the operation, which does not contain
the $\bcopy{n}{i}$ operation) and the \emph{constructive states}
(which label the constructive phase, where no $\nbcopy{n}{i}$
occurs). These sets are further divided into \emph{test states}, which
are reached after a test has been performed (and only then) and which
are the source of no test-labelled transition, and the others. This
transformation can be performed without altering the accepted relation
over stack trees.

\section{Rewriting Graphs of Stack Trees}\label{sec:rew_graph}

In this section, we study the properties of ground stack tree rewriting 
graphs.
Our goal is to show that the graph of any $\Gamma$-labelled GSTRS has a decidable
FO$[\xrightarrow{*}]$ theory. We first state that there exists a
distinguished and
reduced automaton $A$ recognising the derivation relation $\xrightarrow{*}_R$ of $R$, and
then show, following \cite{ColcombetL07}, that there exists a finite-set 
interpretation of $\xrightarrow{*}_R$ and every $\xrightarrow{a}_R$ for 
$(D, a) \in R$ from a graph with decidable WMSO-theory
.

\begin{theorem}
  \label{thm:fo}
  Given a $\Gamma$-labelled GSTRS $R$, $\mathcal{G}_R$ has a
  decidable FO$[ \xrightarrow{*} ]$ theory.
\end{theorem}



To prove this theorem, we show that the graph $\mathcal{H}_R = (V,E)$ with $V = 
\TS_{n}$ and $E = (\xrightarrow{*}_R) \cup \bigcup_{a \in \Gamma} (\xrightarrow{a}_R)$
 obtained by adding the relation $\xrightarrow{*}_R$ to
 $\mathcal{G}_R$ has a decidable FO theory.
 To do so, we show that $\mathcal{H}_R$ is finite-set interpretable
 inside a structure with a decidable WMSO-theory, and conclude using
 Corollary 2.5 of \cite{ColcombetL07}. Thus from Section 5.2 of the
 same article, it follows that the rewriting graphs of GSTRS are in
 the tree-automatic  hierarchy.

Given a $\Gamma$-labelled GSTRS $R$ over $\TS_n$, we choose to interpret 
$\mathcal{H}_R$ inside the \emph{order $n$ Treegraph}
$\Delta^n$ 
over alphabet $\Sigma\cup \{1,2\}$. Each vertex
of this graph is an $n$-stack, and there is an edge
$s\xrightarrow{\theta} s'$ if and only if $s' = \theta(s)$ with $\theta \in
\Ops{n}\cup \Tests{n}$. This graph belongs to the  
$n$-th level of the pushdown hierarchy and has a decidable WMSO
theory\footnotemark. 

\footnotetext{It is in fact a generator of this class of graphs via
  WMSO-interpretations (see \cite{CarayolW03} for additional
  details).}


Given a stack tree $t$ and a position $u \in dom(t)$, we denote by
$\Code{t}{u}$ the $n$-stack 
$\stack{n}{\push{w_0}(t(\varepsilon)),\push{w_1}(t(u_{\leq 
1})),\cdots,\push{w_{ |u|-1 } } 
(t(u_{\leq |u|-1})),t(u)}$, where $\push{w}(s)$ is 
obtained by adding the word $w$ at the top of the top-most 
1-stack in $s$, and $w_i = \#(u_{\leq i}) 
u_{i+1}$. This stack $\Code{t}{u}$ is the encoding of the node at 
position $u$ in $t$. Informally, it is obtained by storing in an $n$-stack  
the sequence of $(n-1)$-stacks labelling nodes from the root of $t$
to position $u$, and adding at the top of each $(n-1)$-stack the number of children 
of the corresponding node of $t$ and the next direction taken to reach node $u$.
Any stack tree $t$ is then encoded by the finite set of $n$-stacks
$X_t = \{\Code{t}{u} \mid u \in fr(t)\}$, i.e. the set of encodings of
its leaves. Observe that this coding is injective.

\begin{example}
 The coding of the stack tree $t$ depicted in Fig. \ref{fig:stack_tree} is:
 
 \begin{tabular}{lll}
  $X_t = $ & $\{$ & 
$\stack{3}{\stack{2}{\stack{1}{aa}\stack{1}{bab21}}\stack{2}{\stack{1}{aa}
\stack{1}{aaa11}}\stack{2}{\stack{1}{ab}}}$,\\
& & 
$\stack{3}{\stack{2}{\stack{1}{aa}\stack{1}{bab22}}\stack{2}{\stack{1}{aa}\stack
{1}{a}\stack{1}{b21}}\stack{2}{\stack{1}{ba}\stack{1}{ba}\stack{1}{b}}}$,\\
& & 
$\stack{3}{\stack{2}{\stack{1}{aa}\stack{1}{bab22}}\stack{2}{\stack{1}{aa}\stack
{1}{a}\stack{1}{b22}}\stack{2}{\stack{1}{abb}\stack{1}{ab}}}\}$
 \end{tabular}

\end{example}

We now represent any relation $S$ between two stack trees as a
WMSO-formula with two free second-order variables, which holds in
$\Delta^n$ over sets $X_s$ and $X_t$ if and only if $(s, t) \in S$.


\begin{proposition}
 Given a $\Gamma$-labelled GSTRS $R$, there exist WMSO-formul\ae{} $\delta, 
\Psi_a$ and $\phi$ such that:
\begin{itemize}
 \item $\Delta_{\Sigma\cup 
\{1,2\}}^n \models \delta(X)$ if and only if $\exists t \in \TS_n, X = 
X_t$,
 \item $\Delta_{\Sigma\cup \{1,2\}}^n \models \Psi_a(X_s,X_t)$ if and only 
if $t\in D(s)$ for some $(D, a) \in R$,
 \item $\Delta_{\Sigma\cup \{1,2\}}^n \models \phi(X_s,X_t)$ if and only if 
$s \xrightarrow{*}_R t$.
\end{itemize}
\end{proposition}

First note that the intuitive idea behind this interpretation is to
only work on those vertices of $\Delta^n$ which are the encoding of
some node in a stack-tree. Formula $\delta$ will distinguish, amongst
all possible finite sets of vertices, those which correspond to the
set of encodings of all leaves of a stack-tree. Formul\ae{} $\Psi_a$ and
$\phi$ then respectively check the relationship through
$\xrightarrow{a}_R$ (resp. $\xrightarrow{*}_R$) of a pair of
stack-trees.
We give here a quick sketch of the formul\ae{} and a glimpse of their
proof of correction.  More details can be found in appendix
\ref{annex:fsi}.

Let us first detail formula $\delta$, which is of the form
\[
\delta(X)  =  \mathrm{OnlyLeaves}(X) \wedge 
\mathrm{TreeDom}(X) \wedge \mathrm{UniqueLabel}(X).
\]
$\mathrm{OnlyLeaves}(X)$ holds if every element of $X$ codes for a
leaf.  $\mathrm{TreeDom}(X)$ holds if the induced domain is the domain
of a tree and the arity of each node is consistent with the elements
of $X$.  $\mathrm{UniqueLabel}(X)$ holds if for every position $u$ in
the induced domain, all elements which include $u$ agree on its
label. 

\smallskip

%

From here on, variables $X$ and $Y$ will respectively stand for the
encoding of some input stack tree $s$ and output stack-tree $t$.
For each $a \in \Gamma$, $\Psi_a(X, Y)$ is the disjunction of a family
of formul\ae{} $\Psi_D(X, Y)$ for each $(D, a) \in R$. Each $\Psi_D$ is
defined by induction over $D$, simulating each basic operations in
$D$, ensuring that they are applied according to their respective
positions, and to a single closed subtree of $s$ (which simply
corresponds to a subset of $X$), yielding $t$.

\smallskip

Let us now turn to formula $\phi$. Since the set of DAGs in $R$ is
finite, it is recognisable by an operation automaton. Since $\Rec$ is
closed under iteration (Cf. Sec. \ref{sec:gstt}), one may build a
distinguished normalised automaton accepting $\xrightarrow{*}_R$. What
we thus really show is that given such an automaton $A$, there exists
a formula $\phi$ such that $\phi(X, Y)$ holds if and only if
$t \in \bar{D}(s)$ for some vector $\bar{D} = D_1, \ldots D_k$ of DAGs
accepted by $A$.
Formula $\phi$ is of the form
\[
\phi(X,Y) = \exists \vec{Z}, \init(X,Y,\vec{Z}) \wedge
\diff(\vec{Z}) \wedge \mathrm{Trans}(\vec{Z}).
\]
%
%
Following a common pattern in automata theory, this formula expresses
the existence of an accepting run of $A$ over some tuple of reduced
DAGs $\bar{D}$, and states that the operation corresponding to
$\bar{D}$, when applied to $s$, yields $t$.
Here, $\vec{Z} = Z_{q_1},\cdots,Z_{q_{|Q_A|}}$ defines a labelling of a subset of 
$\Delta_{\Sigma\cup\{1,2\}}^n$ with the states of the automaton, each element 
$Z_q$ of $\vec{Z}$ representing the set of nodes labelled by a given
control state $q$. 
Sub-formula $\init$ checks that only the elements of $X$ (representing
the leaves of $s$) are labelled by initial states, and only those in
$Y$ (leaves of $t$) are labelled by final states.
$\mathrm{Trans}$ ensures that the whole labelling respects the
transition rules of $A$. For each component $D$ of $\bar{D}$, and
since every basic operation constituting $D$ is applied locally and
has an effect on a subtree of height and width at most $2$, this
amounts to a local consistency check between at most three vertices,
encoding two nodes of a stack tree and their parent node. The relative
positions where basic operations are applied is checked using the sets
in $\vec{Z}$, which represent the flow of control states at each step
of the transformation of $s$ into $t$.
Finally, $\diff$ ensures that no stack is labelled by two states
belonging to the same part (destructive, constructive, testing or
non-testing) of the automaton, thus making sure we simulate a unique
run of $A$. This is necessary to ensure that no spurious run is
generated, and is only possible because $A$ is normalised.

\section{Perspectives}\label{sec:concl}


There are several open questions arising from this work. The first one
is the strictness of the hierarchy, and the question of finding simple
examples of graphs separating each of its levels with the
corresponding levels of the pushdown and tree-automatic hierarchies.
A second interesting question concerns the trace languages of stack
tree rewriting graphs. It is known that the trace languages of higher-order 
pushdown automata are the indexed languages \cite{Caucal96}, that the class 
of languages recognised by automatic structures are the context-sensitive 
languages \cite{Rispal02} and that those recognised by tree-automatic 
structures form the class \textsc{Etime} \cite{Meyer08}. However there is to our 
knowledge no characterisation of the languages recognised by ground tree 
rewriting systems. 
It is not hard to define a 2-stack-tree rewriting 
graph whose path language between two specific vertices is $\{u \shuffle u\mid 
u\in \Sigma^*\}$, which we believe cannot be recognised using tree rewriting 
systems or higher-order pushdown
automata\footnote{$\shuffle$ denotes the shuffle product. For every $u,v \in 
\Sigma^*$ and $a,b \in \Sigma$, $u\shuffle \varepsilon = \varepsilon \shuffle 
u = u$, $au \shuffle bv = a(u\shuffle bv) \cup b(au \shuffle v)$}. 
%
Finally, the model of stack trees can be readily extended to trees
labelled by trees. Future work will include the question of extending
our notion of rewriting and Theorem \ref{thm:fo} to this model.

\bibliographystyle{plain}
\bibliography{stack_tree}


\clearpage
\appendix
\input{appendix}

\end{document}

%% file: appendix.tex
\section{Properties of Operation Automata}

%

%
%
%
%





In this section, we show that $\mathrm{Rec}$ is closed under union, 
intersection, 
iteration and contains the finite sets of operations.

\begin{proposition}\label{prop_union}
Given two automata $A_1$ and $A_2$, there exists an automaton $A$ such that 
$\graphaut{A} = \graphaut{A_1} \cap \graphaut{A_2}$
\end{proposition}

\begin{proof}
We will construct an automaton which witness Prop \ref{prop_union}.
First, we ensure that the two automata are complete by adding a sink state if 
some transitions do not exist.
We construct then the automaton $A$ which is the product automaton of $A_1$ and 
$A_2$:

$Q= Q_{A_1} \times Q_{A_2}$

$I= I_{A_1}\times I_{A_2}$

$F= F_{A_1} \times F_{A_2}$

\begin{tabular}{lll}
$\Delta$ & $=$    & $\{((q_1,q_2),\theta,(q'_1,q'_2)) \mid (q_1,\theta,q'_1) 
\in \Delta_{A_1} \land (q_2,\theta,q'_2)\in \Delta_{A_2}\}$\\
         & $\cup$ & $\{(((q_1,q_2),(q'_1,q'_2)),(q''_1,q''_2)) \mid 
((q_1,q'_1),q''_1)\in \Delta_{A_1} \land ((q_2,q'_2),q''_2)\in 
\Delta_{A_2}\}$\\
         & $\cup$ & $\{((q_1,q_2),((q'_1,q'_2),(q''_1,q''_2))) \mid 
(q_1,(q'_1,q''_1))\in \Delta_{A_1} \land (q_2,(q'_2,q''_2))\in\Delta_{A_2}\}$\\
\end{tabular}

If an operation admits a valid labelling in $A_1$ and in $A_2$, then the 
labelling 
which labels each states by the two states it has in its labelling in $A_1$ and 
$A_2$ is valid.
If an operation admits a valid labelling in $A$, then, restricting it to the 
states of $A_1$ (resp $A_2$), we have a valid labelling in $A_1$ (resp $A_2$).
 \qed
\end{proof}

%

\begin{proposition}
Given two automata $A_1$ and $A_2$, there exists an automaton $A$ such that 
$\graphaut{A} = \graphaut{A_1} \cup \graphaut{A_2}$
\end{proposition}

\begin{proof}
We take the disjoint union of $A_1$ and $A_2$:

$Q = Q_{A_1} \uplus Q_{A_2}$

$I = I_{A_1} \uplus I_{A_2}$

$F = F_{A_1} \uplus F_{A_2}$

$\Delta = \Delta_{A_1} \uplus \Delta_{A_2}$

If an operation admits a valid labelling in $A_1$ (resp $A_2$), it is also a 
valid 
labelling in $A$.
If an operation admits a valid labelling in $A$, as $A$ is a disjoint union of 
$A_1$ 
and $A_2$, it can only be labelled by states of $A_1$ or of $A_2$ (by 
definition, there is no transition between states of $A_1$ and states of $A_2$) 
and then the labelling is valid in $A_1$ or in $A_2$.
 \qed
\end{proof}

%

\medskip

\begin{proposition}
Given an automaton $A$, there exists $A'$ which recognises $\graphaut{A}^*$.
\end{proposition}

\begin{proof}
We construct $A'$.

$Q = Q_A \uplus \{q\}$

$I = I_A \cup \{q\}$

$F = F_A \cup \{q\}$

The set of transition $\Delta$ contains the transitions of $A$ together with 
multiple copies of each transition ending with a state in $F_A$, modified to 
end in a state belonging to $I_A$

\begin{tabular}{lll}
$\Delta$ & $=$    & $\Delta_A$\\
	 & $\cup$ & $\{(q_1,\theta,q_i) \mid q_i \in I_A,\exists q_f \in F_A, 
(q_1,\theta,q_f) \in \Delta_A\}$\\
	 & $\cup$ & $\{((q_1,q_2),q_i) \mid q_i \in I_A, \exists q_f \in F_A, 
((q_1,q_2),q_f) \in \Delta_A\}$\\
	 & $\cup$ & $\{(q_1,(q_2,q_i)) \mid q_i \in I_A, \exists q_f \in F_A, 
(q_1,(q_2,q_f))\in \Delta_A\}$\\
	 & $\cup$ & $\{(q_1,(q_i,q_2)) \mid q_i \in I_A, \exists q_f \in F_A, 
(q_1,(q_f,q_2))\in \Delta_A\}$\\
	 & $\cup$ & $\{(q_1,(q_i,q'_i)) \mid q_i,q'_i \in I_A, \exists q_f,q'_f 
\in F_A, (q_1,(q_f,q'_f))\in \Delta_A\}$\\
\end{tabular}

For every $k\in \mathbb{N}$, if $D \in (\graphaut{A}^k)$, it has a 
valid labelling in $A'$:
The operation $\emptydag$ has a valid labelling because $q$ is initial and 
final. So it is true for $(\graphaut{A}^0)$
If it is true for $(\graphaut{A}^k)$, we take an operation $G$ in 
$(\graphaut{A}^{k+1})$ and decompose it in $D$ of 
$\graphaut{A}$ and $F$ of $\graphaut{A}^k$ (or symmetrically, $D 
\in \graphaut{A}^k$ and $F \in \graphaut{A}^k$), 
such that $G \in D \cdot F$.
The labelling which is the union of some 
valid labellings for $D$ and $F$ and labels the identified nodes 
with the labelling of $F$ (initial states) is valid in $A$.

If an operation admits a valid labelling in $A'$, we can separate several 
parts of the operation, separating on the added transitions, and we obtain a 
collection of operations of $\graphaut{A}$. Then we have a graph in 
$\graphaut{A}^k$ for a given $k$.
Then $\graphaut{A'} = \bigcup_{k \geq 0}\graphaut{A}^k$, then 
$A'$ recognises $\graphaut{A}^*$.
%
%
 \qed
\end{proof}
\medskip

\begin{proposition}
Given an operation $D$, there exists an 
automaton $A$ such that $\graphaut{A} = \{D\}$.
\end{proposition}

\begin{proof}
If $D = (V,E)$, we take:

$Q = V$

$I$ is the set of incoming vertices

$F$ is the set of output vertices

\begin{tabular}{lll}
$\Delta$ & $=$    & $\{(q,\theta,q') \mid (q,\theta,q') \in E\}$\\
	 & $\cup$ & $\{(q,(q',q'')) \mid (q,1,q') \in E \land 
(q,2,q'') \in E\}$\\
	 & $\cup$ & $\{((q,q'),q'') \mid (q,1,q'') \in E \land 
(q',2,q'') \in E\}$\\
\end{tabular}

The recognised connected part is $D$ by construction.
 \qed
\end{proof}

\section{Normalised Automata}

\begin{definition}

An automaton is normalised if all its recognised operations are reduced.
\end{definition}

\begin{theorem}
Given an operation automaton with tests, there exists a distinguished 
normalised operation automaton with tests which accepts the same language.
\end{theorem}

\begin{proof}
The first thing to remark is that if we don't have any tree transitions, we 
have 
a higher-order stack automaton as in \cite{Carayol05} and that the notions of 
normalised automaton coincide. The idea is thus to separate the automaton in 
two parts, one containing only tree transitions and the other stack 
transitions, to normalise each part separately and then to remove the useless 
transitions used to separate the automaton.
\bigskip

%
%
%
%

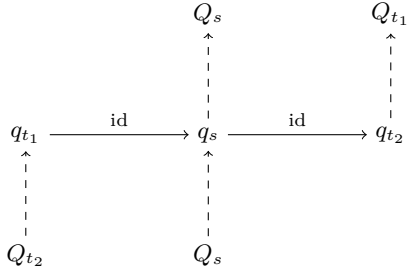
\begin{figure}
\begin{center}
\begin{tikzpicture} [scale=.4]
\node(1) at (0,-4) {$q_{t_1}$};
\node(2) at (6,-4) {$q_s$};
\node(3) at (12,-4) {$q_{t_2}$};
\node(9) at (0,-8) {$Q_{t_2}$};

\node(12) at (6,0) {$Q_s$};
\node(13) at (6,-8) {$Q_s$};

\node(15) at (12,0) {$Q_{t_1}$};

\draw[->] (1) to node[midway,above]{\scriptsize{$\Id$}} (2);
\draw[->] (2) to node[midway,above]{\scriptsize{$\Id$}} (3);

\draw[->,dashed] (9) to (1);
\draw[->,dashed] (13) to (2);
\draw[->,dashed] (2) to (12);
\draw[->,dashed] (3) to (15);

\end{tikzpicture}
\end{center}
\caption{Step 1: The splitting of a state $q$}\label{fig_step1}
\end{figure}

\paragraph*{Step 1:}
In this transformation, we will use a new special basic operation: $\Id$ such 
that its associated operation $D_{\Id}$ is the following DAG: $V_{D_{\Id}} = 
\{x,y\}$ and $E_{D_{\Id}} = 
\{(x,\Id,y)\}$. For every stack tree $t$ and any integer $i \leq |\fr(t)|$, 
$\Id_{(i)}(t) = 
t$.
We will use this operation to separate our DAGs in several parts linked with 
$\Id$ operations, and will remove them at the end of the transformation. We 
suppose that we start with an automaton without such $\Id$ transitions.

We begin by splitting the set of control states of the automaton into three 
parts. We create three copies of $Q$: 
\begin{itemize}
 \item $Q_s$ which are the sources and targets of all the stack 
transitions, target of $\Id$ transitions from $Q_{t_1}$ and source of 
$\Id$-transitions to $Q_{t_2}$.
 \item $Q_{t_1}$ which are the targets of all the tree transitions and the 
sources of $\Id$-transitions to $Q_s$.
 \item $Q_{t_2}$ which are the sources of all the tree transitions and the 
targets of $\Id$-transitions from $Q_s$.
\end{itemize}


The idea of what we want to obtain is depicted in Fig. \ref{fig_step1}.

Formally, we replace the automaton $A = (Q,I,F,\Delta)$ by $A_1 = 
(Q',I',F',\Delta')$ with:

\begin{tabular}{lll}
$Q'$ & $ =$ & $\{q_{t_1},q_{t_2},q_s \mid q\in Q\}$\\

$I'$ & $=$ & $\{q_s \mid q\in I\}$\\

$F'$ & $=$ & $\{q_s \mid q\in F\}$\\

$\Delta$ &$ = $&$ \{(q_s,\theta,q'_s)\mid (q,\theta,q') \in \Delta\}$\\
 &$ \cup $&$ \{(q_{t_2},(q'_{t_1},q''_{t_1}))\mid (q,(q',q'')) \in \Delta\}$\\
 &$ \cup $&$ \{((q_{t_2},q'_{t_2}),q''_{t_1})\mid ((q,q'),q'') \in \Delta\}$\\
 &$ \cup $&$ \{(q_{t_2},\bcopy{n}{1},q'_{t_1}) \mid (q,\bcopy{n}{1},q')\in 
\Delta\}$\\
 &$ \cup $&$ \{(q_{t_2},\nbcopy{n}{1},q'_{t_1}) \mid (q,\nbcopy{n}{1},q')\in 
\Delta\}$\\
 &$ \cup $&$ \{(q_{t_1},\Id,q_s),(q_s,\Id,q_{t_2})\mid q\in Q\}$
\end{tabular}

where for every $q\in Q$, $q_{t_1},q_{t_2},q_s$ are fresh states.

\begin{lemma}
$A$ and $A_1$ recognise the same relation. 
\end{lemma}

\begin{proof}
To prove this lemma, we prove that for every operation $D$ recognised by 
$A$, there is an operation $D'$ recognised by $A_1$ such that $R_D = 
R_{D'}$, and vice versa.

Let us take $D$ recognised by $A$. We prove, by induction on the structure 
of $D$ that we can construct $D'$ such that $R_D = R_{D'}$ 
and for every labelling $\rho_D$ of $D$ consistent with $\Delta$, 
with $I_{D}$ labelled by $\vec{q}$ and $O_{D}$ by $\vec{q'}$, 
there exists $\rho_D'$ a labelling of $D'$ consistent with $\Delta'$ 
such that $I_{D'}$ is labelled by $\vec{q_s}$ and $O_{D'}$ by 
$\vec{q'_s}$.

If $D = \emptydag$, we take $D' = \emptydag$. We have $R_D = 
R_{D'}$. For every labelling $\rho_D$ which labels the unique node of 
$D$ by $q$, we take $\rho_{D'}$ which labels the unique node of 
${D'}$ by $q_s$. These labellings are consistent by $\Delta$ and 
$\Delta'$, by vacuity.

Suppose now that we have $F$ and $F'$ such that for 
every labelling $\rho_F$ we can define a labelling $\rho_{F'}$ 
satisfying the previous condition. Let us consider the following cases:

\begin{itemize}
 \item $D = (F \cdot_{1,1} D_\theta) \cdot_{1,1} G$, 
for $\theta \in \{\bcopy{n}{1},\nbcopy{n}{1}\}$. We call $x$ the output node of 
$F$ and $y$ the input node of $G$. We have $V_{D} = 
V_{F} \cup V_{G}$ and $E_{D} = E_{F} \cup 
E_{G} \cup \{x\xrightarrow{\theta}y\}$.

By induction hypothesis, we consider $F'$ and $G'$, and construct 
${D'} = ((({F'} \cdot_{1,1} D_{\Id}) \cdot_{1,1} D_\theta) 
\cdot_{1,1} D_{\Id}) \cdot_{1,1} {G'} $, with $V_{{D'}} = 
V_{{F'}}\cup V_{{G'}} \cup \{x'_1,x'_2\}$ and 
$E_{D'} = E_{F'} \cup E_{G'}\cup \{x' 
\xrightarrow{\Id} x'_1,x'_1 \xrightarrow{\theta} x'_2,x'_2 \xrightarrow{\Id} 
y'\}$, where $x'$ is the output node of $F'$ and $y'$ the input node of 
$G'$.

We take $\rho_D$ a labelling of $D$ and $\rho_F$ (resp. 
$\rho_G$) its restriction to $F$ (resp. $G$). We have 
$\rho_D(x) = q$ and $\rho_D(y) = q'$. By induction hypothesis, we 
consider $\rho_{F'}$ (resp. $\rho_{G'}$) the corresponding labelling 
of $F'$ (resp. $G'$), with $\rho_{F'}(x') = q_s$ (resp. 
$\rho_{G'}(y') = q'_s$). Then, we construct $\rho_{D'} = 
\rho_{F'} \cup \rho_{G'} \cup \{x'_1 \rightarrow q_{t_2}, x'_2 
\rightarrow q'_{t_1}\}$.

 As $\rho_D$ is consistent with $\Delta$, $(q,\theta,q')$ is in $\Delta$, 
then by construction $(q_{t_2},\theta,q'_{t_1})$ is in $\Delta'$. We have also 
$(q_s,\Id,q_{t_2})$ and $(q'_{t_1},\Id,q'_{s})$ are in $\Delta'$. Then, 
$\rho'_D$ is consistent with $\Delta'$.

To prove that $R_D = R_{D'}$, we just have to remark that, from the 
definition of application of operation, we have for every stack tree $t$ and 
integer $i$, we have $D'_{(i)}(t) = G'_{(i)}( \Id_{(i)}( 
\theta_{(i)}( \Id_{(i)}( F'_{(i)}(t))))) = 
G_{(i)}(\theta_{(i)}(F_{(i)}(t))) = D_{(i)}(t)$.
\smallskip

The other cases being similar, we just give $D'$ and $\rho_{D'}$ 
and leave the details to the reader.

 \item $D = (F \cdot_{1,1} D_\theta) \cdot_{1,1} G$, 
for $\theta \in \Ops{n-1} \cup \Tests{n-1}$. We call $x$ the output node of 
$F$ and $y$ the input node of $G$. We have $V_D = 
V_F \cup V_G$ and $E_D = E_F \cup 
E_G \cup \{x\xrightarrow{\theta}y\}$.

By induction hypothesis, we consider $F'$ and $G'$, and construct 
$D' = (F' \cdot_{1,1} \theta) \cdot_{1,1} G' $, 
with $V_{D'} = V_{F'}\cup V_{G'}$ 
and $E_{D'} = E_{F'} \cup E_{G'}\cup \{x'  
\xrightarrow{\theta} y'\}$, where $x'$ is the output node of $F'$ and $y'$ 
the input node of $G'$.

We take $\rho_D$ a labelling of $D$ and $\rho_F$ (resp. 
$\rho_G$) its restriction to $F$ (resp. $G$). We have 
$\rho_D(x) = q$ and $\rho_D(y) = q'$. By induction hypothesis, we 
consider $\rho_{F'}$ (resp. $\rho_{G'}$) the corresponding labelling 
of $F'$ (resp. $G'$), with $\rho_{F'}(x') = q_s$ (resp. 
$\rho_{G'}(y') = q'_s$). Then, we construct $\rho_{D'} = 
\rho_{F'} \cup \rho_{G'}$.

 \item $D = ((F \cdot_{1,1} D_{\bcopy{n}{2}}) \cdot_{2,1} 
H) \cdot_{1,1} G$. We call $x$ the output node of 
$F$, $y$ the input node of $G$ and $z$ the input node of 
$H$. We have $V_D = 
V_F \cup V_G \cup V_H$ and $E_D = 
E_F \cup E_G \cup E_H \cup 
\{x\xrightarrow{1}y,x\xrightarrow{2}z\}$.

By induction hypothesis, we consider $F'$, $G'$ and $H'$, and 
construct $D' = (((((F \cdot_{1,1} D_{\Id}) D_{\bcopy{n}{2}}) 
\cdot_{2,1} D_{\Id}) \cdot_{2,1} H) \cdot_{1,1} D_{\Id}) \cdot_{1,1} 
G$, with $V_{D'} = 
V_{F'}\cup V_{G'} \cup V_{H'} \cup 
\{x'_1,x'_2,x'_3\}$ and 
$E_{D'} = E_{F'} \cup E_{G'}\cup E_{H'} 
\{x' 
\xrightarrow{\Id} x'_1,x'_1 \xrightarrow{1} x'_2,x'_1\xrightarrow{2} x'_3,x'_2 
\xrightarrow{\Id} y', x'_3 \xrightarrow{\Id} z'\}$, where $x'$ is the output 
node of $F'$, $y'$ the input node of $G'$ and $z'$ the input node of 
$H'$.

We take $\rho_D$ a labelling of $D$ and $\rho_F$ (resp. 
$\rho_G$, $\rho_H$) its restriction to $F$ (resp. 
$G$, $H$). We have $\rho_D(x) = q$, $\rho_D(y) = q'$ 
and $\rho_D(z) = q''$. By induction hypothesis, we 
consider $\rho_{F'}$ (resp. $\rho_{G'}$,$\rho_{H'}$) the 
corresponding labelling of $F'$ (resp. $G'$,$H'$), with 
$\rho_{F'}(x') = q_s$ (resp. $\rho_{G'}(y') = q'_s$, 
$\rho_{H'}(z') = q''_s$). Then, we construct $\rho_{D'} = 
\rho_{F'} \cup \rho_{G'} \cup \rho_{H'} \cup \{x'_1 \rightarrow 
q_{t_2}, x'_2 \rightarrow q'_{t_1},x'_3 \rightarrow q''_{t_1}\}$.

 \item $D = (F \cdot_{1,1} (G \cdot_{1,2} 
D_{\nbcopy{n}{2}})) \cdot_{1,1} H$. We call $x$ the output node of 
$F$, $y$ the output node of $G$ and $z$ the input node of 
$H$. We have $V_D = V_F \cup V_G \cup 
V_H$ and $E_D = E_F \cup E_G \cup E_H \cup 
\{x\xrightarrow{\bar{1}}z,y\xrightarrow{\bar{2}}z\}$.

By induction hypothesis, we consider $F'$, $G'$ and $H'$, and 
construct $D' = (((F \cdot_{1,1} D_{\Id}) \cdot_{1,1} ((
G \cdot_{1,1} D_{\Id}) \cdot_{1,2} D_{\nbcopy{n}{2}})) \cdot_{1,1} 
D_{\Id}) \cdot_{1,1} H$, with $V_{D'} = V_{F'}\cup 
V_{G'} \cup V_{H'} \cup \{x'_1,x'_2,x'_3\}$ and 
$E_{D'} = E_{F'} \cup E_{G'}\cup E_{H'} \{x' 
\xrightarrow{\Id} x'_1,y' \xrightarrow{\Id} x'_2,x'_1 \xrightarrow{\bar{1}} 
x'_3,x'_2\xrightarrow{\bar{2}} x'_3, x'_3 \xrightarrow{\Id} z'\}$, where $x'$ 
is 
the output node of $F'$, $y'$ the input node of $G'$ and $z'$ the 
input 
node of $H'$.

We take $\rho_D$ a labelling of $D$ and $\rho_F$ (resp. 
$\rho_G$, $\rho_H$) its restriction to $F$ (resp. 
$G$, $H$). We have $\rho_D(x) = q$, $\rho_D(y) = q'$ 
and $\rho_D(z) = q''$. By induction hypothesis, we 
consider $\rho_{F'}$ (resp. $\rho_{G'}$, $\rho_{H'}$) the 
corresponding labelling of $F'$ (resp. $G'$, $H'$), with 
$\rho_{F'}(x') = q_s$ (resp. $\rho_{G'}(y') = q'_s$, 
$\rho_{H'}(z') = q''_s$). Then, we construct $\rho_{D'} = 
\rho_{F'} \cup \rho_{G'} \cup \rho_{H'} \cup \{x'_1 \rightarrow 
q_{t_2}, x'_2 \rightarrow q'_{t_2},x'_3 \rightarrow q''_{t_1}\}$.

 \item $D = (((((F \cdot_{1,1} D_{\bcopy{n}{2}}) \cdot_{2,1} 
H) \cdot_{1,1} G) \cdot_{1,1} D_{\nbcopy{n}{2}}) \cdot_{1,1} 
K$. We call $x$ the output node of $F$, $y_1$ the input 
node of $G$ and $y_2$ its output node, $z_1$ the input node of 
$H$ and $z_2$ its output node and $w$ the input node of $K$. We have 
$V_D = V_F \cup V_G \cup V_H \cup 
V_K$ and $E_D = E_F \cup E_G \cup 
E_H \cup E_K \cup 
\{x\xrightarrow{1}y_1,x\xrightarrow{2}z_1,y_2\xrightarrow{\bar{1}}t,
z_2\xrightarrow { \bar { 2 }}t\}$.

By induction hypothesis, we consider $F'$, $G'$, $H'$ and 
$K'$, and construct $D' = ((((((F' \cdot_{1,1} D_{\Id}) 
\cdot_{1,1} D_{\bcopy{n}{2}}) \cdot_{2,1} (D_{\Id} \cdot_{1,1} H')) 
\cdot_{1,1} (D_{\Id} \cdot_{1,1} G')) \cdot_{1,1} D_{\nbcopy{n}{2}}) 
\cdot_{1,1} D_{\Id}) \cdot_{1,1} K'$, with $V_{D'} = 
V_{F'}\cup V_{G'} \cup V_{H'} \cup V_{K'} \cup 
\{x'_1,x'_2,x'_3,x'_4,x'_5,x'_6\}$ and $E_{D'} = E_{F'} \cup 
E_{G'}\cup E_{H'} \cup E_{K'} \{x' \xrightarrow{\Id} 
x'_1,x'_1 \xrightarrow{1} x'_2,x'_1\xrightarrow{2} x'_3,x'_2 \xrightarrow{\Id} 
y'_1, x'_3 \xrightarrow{\Id} z'_1,y'_2 \xrightarrow{\Id} 
x'_4,z'_2 \xrightarrow{\Id} x'_5, x'_4 \xrightarrow{\bar{1}} 
x'_6,x'_5\xrightarrow{\bar{2}} x'_6,x'_6 \xrightarrow{\Id} t'\}$, where $x'$ is 
the output node of $F'$, $y'_1$ the 
input node of $G'$, $y'_2$ its output node, $z'_1$ the input node of 
$H'$, $z'_2$ its output node and $t'$ the input node of $K'$.

We take $\rho_D$ a labelling of $D_D$ and $\rho_F$ (resp. 
$\rho_G$, $\rho_H$, $\rho_K$) its restriction to $F$ (resp. 
$G$, $H$, $K$). We have $\rho_D(x) = q$, 
$\rho_D(y_1) = q'$, $\rho_D(z_1) = q''$, $\rho_D(y_2) = r'$, 
$\rho_D(z_2) = r''$ and $\rho_D(t) = r''$ . By induction hypothesis, 
we consider $\rho_{F'}$ (resp. $\rho_{G'}$, $\rho_{H'}$, 
$\rho_{K'}$) the corresponding labelling of $F'$ (resp. 
$G'$, $H'$, $K'$), with $\rho_{F'}(x') = q_s$ (resp. 
$\rho_{G'}(y'_1) = q'_s$, $\rho_{H'}(z'_1) = 
q''_s$, $\rho_{G'}(y'_2) = r'_s$, $\rho_{H'}(z'_2) = r''_s$, 
$\rho_{K'}(t') = r''_s$). Then, we construct $\rho_{D'} = 
\rho_{F'} \cup \rho_{G'} \cup \rho_{H'} \cup \{x'_1 \rightarrow 
q_{t_2}, x'_2 \rightarrow q'_{t_1},x'_3 \rightarrow q''_{t_1},x'_4 \rightarrow 
r_{t_2}, x'_5 \rightarrow r'_{t_2},x'_6 \rightarrow r''_{t_1}\}$.
\end{itemize}

To do the other direction, we take $D'$ recognised by $A_1$ and show that 
we can construct $D$ recognised by $A$ with $R_D = R_{D'}$ by an 
induction on the structure of $D'$ similar to the previous one (for each 
$\Id$ transition, we do not modify the constructed DAG and for all other 
transition, we add them to the DAG). All the arguments are similar to the 
previous proof, so we let the reader detail it.
\qed
\end{proof}

\bigskip

We start by normalising the tree part of the automaton. To do so, we just have 
to prevent the automaton to recognise DAGs which contain $((D_{\bcopy{n}{2}} 
\cdot_{1,1} F_1) \cdot_{2,1} F_2) \cdot_{1,1} 
D_{\nbcopy{n}{2}}$, or $(D_{\bcopy{n}{1}} \cdot_{1,1} F) \cdot_{1,1} 
D_{\nbcopy{n}{1}}$ as a subDAG. Such a subDAG will be called a 
bubble. However, we do not want to modify the recognised relation. We will do 
it in two steps: first we allow the automaton to replace the bubbles with 
equivalent tests (after remarking that a bubble can only be a test) in any 
recognised DAG (step 2), and then by ensuring that there won't be any 
$\nbcopy{n}{i}$ transition below the first $\bcopy{n}{j}$ transition (step 3).


\paragraph*{Step 2:}
Let $A_1 = (Q,I,F,\Delta)$ be the automaton obtained after step 1.
Given two states $q_1,q_2$, we denote by $L_{A_{q_1,q_2}}$ the set $\{s \in 
\Stacks_{n-1} \mid \exists D\in \graphops{}(A_1), D_{(1)}(s) = s\}$ 
where 
$A_{q_1,q_2}$ is a copy of $A_1$ in which we take $q_1$ as the unique initial 
state and $q_2$ as the unique final state. In other words, $L_{A_{q_1,q_2}}$ is 
the set of $(n-1)$-stacks such that the trees with one node labelled by this 
stack remains unchanged by an operation recognised by $A_{q_1,q_2}$.
We define $A_2 = (Q,I,F,\Delta')$ with

\begin{tabular}{lll}
$\Delta'$ & $=$ & $\Delta$\\
& $\cup$ & $\{(q_s,T_{L_{A_{r_s,r'_s}} \cap L_{A_{s_s,s'_s}}},q'_s) 
\mid (q_{t_2},(r_{t_1},s_{t_1})), ((r'_{t_2},s'_{t_2}),q'_{t_1}) \in \Delta 
\}$\\
& $\cup$ & $\{(q_s,T_{L_{r_s,s'_s}},q'_s \mid (q_{t_2},\bcopy{n}{1},r_{t_1}), 
(r'_{t_2},\nbcopy{n}{1},q'_{t_1}) \in \Delta \}$
\end{tabular}

%
%
The idea of the construction is depicted in Fig. \ref{fig_step2}.

We give the following lemma for the binary bubble. The case of the unary bubble 
is very similar and thus if left to the reader.

\begin{lemma}
  Let $C_1 = (Q_{C_1},\{i_{C_1}\},\{f_{C_1}\},\Delta_{C_1})$ and $C_2
  = (Q_{C_2},\{i_{C_2}\},\{f_{C_2}\},\Delta_{C_2})$ be two automata
  recognising DAGs without tree operations. The two automata $B_1 =
  (Q_1, I, F, \Delta_1)$ and $B_2 = (Q_2, I, F, \Delta_2)$, with $I =
  \{q_1\}$, $F = \{q_2\}$, $Q_1 = \{q_1,q_2\}$, $\Delta_1 =
  \{(q_1,\Test{L_{C_1} \cap L_{C_2}},q_2)\}$, $Q_2 = \{q_1,q_2\} \cup
  Q_{C_1} \cup Q_{C_2}$ and $\Delta_2 =
  \{(q_1,(i_{C_1},i_{C_2})),((f_{C_1},f_{C_2}),q_2)\} \cup
  \Delta_{C_1} \cup \Delta_{C_2}$ recognise the same relation.
\end{lemma}

\begin{proof}
 An operation $D$ recognised by $B_2$ is of the form $D = 
D_{\bcopy{n}{2}} \cdot_{1,1} (F_1 \cdot_{1,1} (F_2 \cdot_{2,2} 
D_{\nbcopy{n}{2}}))$, where $F_1$ is recognised by $C_1$ and $F_2$ by 
$C_2$.
We have:
\begin{align*}
  D_{(i)}(t) & = {\nbcopy{n}{2}}_{(i)} ({F_1}_{(i)} (
  {F_2}_{(i+1)} ( {\bcopy{n}{2}}_{(i)} (t))))
  \\
  & = {\nbcopy{n}{2}}_{(i)} ({F_1}_{(i)} ( {F_2}_{(i+1)}
  (t\cup \{u_i1 \mapsto t(u_i), u_i2 \mapsto t(u_i)\})))
  \\
  & = {\nbcopy{n}{2}}_{(i)} (t\cup \{u_i1 \mapsto F_1(t(u_i)),
  u_i2 \mapsto F_2(t(u_i))\}).
\end{align*}
So this operation is defined if and only if $F_1(t(u_i)) = F_2(t(u_i)) 
= t(u_i)$. In this case, $D_i(t) = t$.
Thus, $B_2$ accepts only operations which are tests, and these tests are the 
intersection of the tests recognised by $C_1$ and $C_2$. So the relation 
recognised by $B_2$ is exactly the relation recognised by $\Test{L_{C_1} \cap 
L_{C_2}}$, which is the only operation recognised by $B_1$.
 \qed
\end{proof}

We have the following corollary as a direct consequence of this lemma.

\begin{corollary}
 $A_1$ and $A_2$ recognises the same relation.
\end{corollary}

 Indeed, all the new operations recognised do not modify the relation 
recognised by the automaton as each test was already present in the DAGs 
containing a bubble.

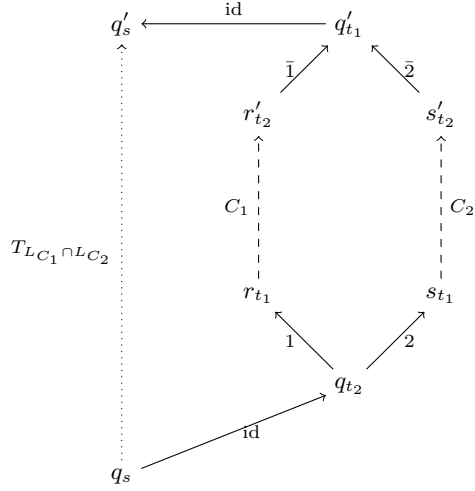
\begin{figure}
\begin{center}
\begin{tikzpicture} [scale=.6]
\node(1) at (0,0) {$q_{t_2}$};
\node(2) at (-2,2) {$r_{t_1}$};
\node(3) at (2,2) {$s_{t_1}$};
\node(5) at (-2,6) {$r'_{t_2}$};
\node(6) at (2,6) {$s'_{t_2}$};
\node(9) at (0,8) {$q'_{t_1}$};

\node(8') at (-5,-2) {$q_s$};
\node(4) at (-5,8) {$q'_s$};


\draw[->] (1) to node[midway,left]{\scriptsize{$1$}} (2);
\draw[->] (1) to node[midway,right]{\scriptsize{$2$}} (3);
\draw[->,dashed] (2) to node[midway,left]{\scriptsize{$C_1$}} (5);
\draw[->,dashed] (3) to node[midway,right]{\scriptsize{$C_2$}} (6);
\draw[->] (5) to node[midway,left]{\scriptsize{$\bar{1}$}} (9);
\draw[->] (6) to node[midway,right]{\scriptsize{$\bar{2}$}} (9);

\draw[->] (8') to node[midway,right]{\scriptsize{$\Id$}} (1);
\draw[->,dotted] (8') to node[midway,left]{\scriptsize{$\Test{L_{C_1} \cap 
L_{C_2}}$}} 
(4);
\draw[->] (9) to node[midway,above]{\scriptsize{$\Id$}} (4);


%
\end{tikzpicture}
\end{center}
\caption{Step 2: The added test transition to shortcut the 
bubble is depicted with a dotted line}\label{fig_step2}
\end{figure}

\paragraph*{Step 3:}
Suppose that $A_2 = (Q,I,F,\Delta)$ is the automaton obtained after step 2.
We now want to really forbid these bubbles. To do so, we split the 
control states automaton in two parts: We create 2 copies of $Q$:

\begin{itemize}
 \item $Q_d$ which are target of no $\bcopy{n}{d}$ transition,
 \item $Q_c$ which are source of no $\nbcopy{n}{d}$ transition.
\end{itemize}

We construct $A_3= (Q',I',F',\Delta')$ with:

\begin{tabular}{lll}
$Q'$ & $=$ & $\{q_d,q_c \mid q \in Q\}$\\

$I'$ & $=$ & $\{q_d,q_c\mid q\in I\}$\\

$F'$ & $=$ & $\{q_d,q_c\mid q\in F\}$\\

$\Delta' $ &$=$ &$ \{(q_d,\theta,q'_d),(q_c,\theta,q'_c)\mid (q,\theta,q')\in 
\Delta, \theta \in \Ops{n-1} \cup \Tests{n-1} \cup \{\Id\}\}$\\
& $\cup $ & $\{((q_d,q'_d),q''_d) \mid ((q,q'),q'') \in \Delta\}$\\
& $\cup $ & $\{(q_d,\nbcopy{n}{1},q'_d) \mid (q,\nbcopy{n}{1},q')\in\Delta\}$\\
& $\cup $ & $\{(q_c,(q'_c,q''_c)),(q_d,(q'_c,q''_c))\mid (q,(q',q''))\in 
\Delta\}$\\
& $\cup $ & $\{(q_c,\bcopy{n}{1},q'_c),(q_d,\bcopy{n}{1},q'_c) \mid 
(q,\bcopy{n}{1},q')\in\Delta\}$
\end{tabular}


\begin{lemma}
 $A_2$ and $A_3$ recognise the same relation
\end{lemma}

\begin{proof}
$A_3$ recognises the operations recognised by $A_2$ which contain no bubble. 
Indeed, every labelling of such an operation in $A_2$ can be modified to be a 
labelling in $A_3$ (left to the reader). Conversely, each operation recognised 
by $A_3$ is recognised by $A_2$.

Let us take $D$ recognised by $A_2$ which contains at least one bubble. 
Suppose that $D$ contains a bubble $F$ and that $D = 
D[F]_x$ where $D$ is a DAG with one bubble less and we obtain $D$ 
by replacing the node $x$ by $F$ in $D$. From step 2, there exist 
four states of $A_2$, $r_s,r'_s,s_s,s'_s$ such that $G = 
D[\Test{L_{A_{r_s,r'_s}} \cap L_{A_{s_s,s'_s}}}]_x$ is recognised by $A_2$. 
Then $R_D \subseteq R_G$, and $G$ has one less bubble than 
$D$.

Iterating this process, we obtain an operation $D'$ without any bubble 
such that $R_D \subseteq R_{D'}$ and $D'$ is recognised by 
$A_2$. As it contains no bubble, it is also recognised by $A_3$.


Then every relation recognised by an operation with bubbles is already included 
in the relation recognised by an operation without bubbles. Then $A_2$ and 
$A_3$ recognise the same relation.
 \qed
\end{proof}

We call the destructive part the restriction $A_{3,d}$ of $A_3$ to $Q_d$ and the
constructive part its restriction $A_{3,c}$ to $Q_c$.

\paragraph*{Step 4:}
We consider an automaton $A_3$ obtained after the previous step. Observe that 
in the two previous steps, we did not modify the separation between $Q_{t_1}$, 
$Q_{t_2}$ and $Q_s$. We call $A_{3,s}$ the restriction of $A_3$ to $Q_s$.

We now want to normalise $A_{3,s}$. As this part of the automaton only contains 
transitions labelled by operations of $\Ops{n-1} \cup \Tests{n-1}$, we can 
consider it as an 
automaton over higher-order stack operations. So we will use the process of 
normalisation over higher-order stack operations defined in \cite{Carayol05}. 
For each pair $(q_s,q'_s)$ of states in $Q_s$, we construct the normalised 
automaton $A_{q_s,q'_s}$ of $A'$ where $A'$ is a copy of $A_{3,s}$ where 
$I_{A'} = 
\{q_s\}$ and $F_{A'}=\{q'_s\}$. We suppose that these automata are 
distinguished, i.e. that states of $I_{A_{q_s,q'_s}}$ are target of no 
transitions and states of $F_{A_{q_s,q'_s}}$ are source of no transitions. We 
moreover suppose that it is not possible to do two test transitions in a row 
(this is not a strong supposition because such a sequence would not be 
normalised, but it is worth noticing it).

We replace $A_{3,s}$ with the union of all the $A_{q_s,q'_s}$: we define $A_4 = 
(Q',I',F',\Delta')$:

\begin{tabular}{lll}
$Q'$ & $=$ & $Q_{t_1} \cup Q_{t_2} \cup \bigcup_{q_s,q'_s} Q_{A_{q_s,q'_s}}$\\

$I'$ & $=$ & $\bigcup_{q_s\in I,q'_s\in Q_s} I_{A_{q_s,q'_s}}$\\

$F'$ & $=$ & $\bigcup_{q_s\in Q_s,q'_s\in F} F_{A_{q_s,q'_s}}$\\

$\Delta'$ & $=$ & 
$\{K \in \Delta \mid K=(q,(q',q'')) \vee K=((q,q'),q'') \vee K= 
(q,\bcopy{n}{1},q')$\\
& & $\vee K= (q,\nbcopy{n}{1},q')\}$\\
& $\cup$ & $\bigcup_{q_s,q'_s \in Q_s} \Delta_{A_{q_s,q'_s}}$\\
& $\cup$ & $\{(q_{t_1},\Id,i) \mid (q_{t_1},\Id,q'_s) \in \Delta, i \in 
\bigcup_{q''_s\in Q} I_{A_{q'_s,q''_s}}\}$\\
& $\cup$ & $\{(f,\Id,q_{t_2}) \mid (q'_s,\Id,q_{t_2}) \in \Delta, f \in 
\bigcup_{q''_s\in Q} F_{A_{q''_s,q'_s}}\}$\\
& $\cup$ & $\{(q_{t_1},\Id,f) \mid (q_{t_1},\Id,q'_s) \in \Delta, f \in 
\bigcup_{q''_s\in Q} F_{A_{q''_s,q'_s}}\}$\\
& $\cup$ & $\{(i,\Id,q_{t_2}) \mid (q'_s,\Id,q_{t_2}) \in \Delta, i \in 
\bigcup_{q''_s\in Q} I_{A_{q'_s,q''_s}}\}$
\end{tabular}

\begin{lemma}
 $A_3$ and $A_4$ recognise the same relation.
\end{lemma}

\begin{proof}
For every operation $D$ recognised by $A_3$, we can construct $D'$ by 
replacing each sequence of $\Ops{n-1} \cup \Tests{n-1}$ operations by their 
reduced sequence, 
which is recognised by $A_4$ and define the same relation. The details are left 
to the reader.

Conversely, for every $D'$ recognised by $A_4$, we can construct $D$ 
recognised by $A_3$ which define the same relation, by replacing every 
reduced sequence of $\Ops{n-1}\cup\Tests{n-1}$ operations by a sequence of 
$\Ops{n-1}\cup\Tests{n-1}$ 
operations defining the same relation such that $D$ is recognised by 
$A_3$. We leave the details to the reader.
 \qed
\end{proof}





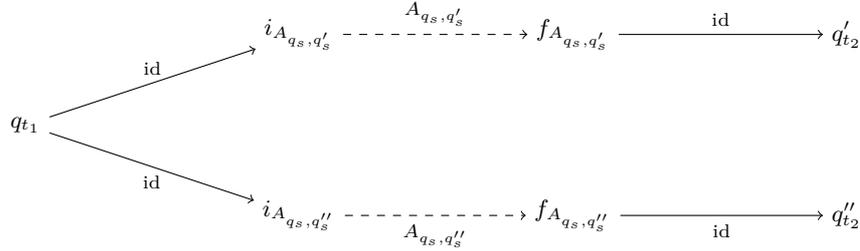
\begin{figure}\label{fig_step4}
\begin{center}
\begin{tikzpicture} [scale=.6]
\node(1) at (0,0) {$q_{t_1}$};
\node(2) at (6,2) {$i_{A_{q_s,q'_s}}$};
\node(2') at (6,-2) {$i_{A_{q_s,q''_s}}$};
\node(3) at (12,2) {$f_{A_{q_s,q'_s}}$};
\node(3') at (12,-2) {$f_{A_{q_s,q''_s}}$};
\node(4) at (18,2) {$q'_{t_2}$};
\node(4') at (18,-2) {$q''_{t_2}$};

\draw[->] (1) to node[midway,above]{\scriptsize{$\Id$}} (2);
\draw[->] (1) to node[midway,below]{\scriptsize{$\Id$}} (2');
\draw[->,dashed] (2) to node[midway,above]{\scriptsize{$A_{q_s,q'_s}$}} (3);
\draw[->,dashed] (2') to node[midway,below]{\scriptsize{$A_{q_s,q''_s}$}} (3');
\draw[->] (3) to node[midway,above]{\scriptsize{$\Id$}} (4);
\draw[->] (3') to node[midway,below]{\scriptsize{$\Id$}} (4');
\end{tikzpicture}
\end{center}
\caption{Step 4: The splitting of the stack part of the automaton}
\end{figure}

\paragraph*{Step 5:}
We now have a normalised automaton, except that we have $\Id$ transitions. We 
remove them by a classical saturation mechanism.
Observe that in all the previous steps, we never modified the separation 
between $Q_{t_1},Q_s$ and $Q_{t_2}$, so that all $\Id$ transitions are from 
$Q_{t_1}$ to $Q_s$ and from $Q_s$ to $Q_{t_2}$.
We take $A_4 = (Q,I,F,\Delta)$ obtained after the previous step. We construct 
$A_5 = (Q',I',F',\Delta')$ with $Q' = Q_s$, $I' = I$, $F' = F$ and
\begin{align*}
\Delta' & = \Delta \setminus \{(q,\Id,q') \in \Delta\}\\
& \cup \{(q_s,\bcopy{n}{1},q'_s) \mid \exists q''_{t_2},q'''_{t_1},
(q''_{t_2},\bcopy{n}{1},q'''_{t_1}), (q'''_{t_1},\Id,q'_s), 
(q_s,\Id,q''_{t_2})\in \Delta\}\\
& \cup \{(q_s,\nbcopy{n}{1},q'_s) \mid \exists q''_{t_2},q'''_{t_1},
(q''_{t_2},\nbcopy{n}{1},q'''_{t_1}), (q'''_{t_1},\Id,q'_s), 
(q_s,\Id,q''_{t_2})\in \Delta\}\\
& \cup \{(q_s,(q'_s,q''_s)) \mid \exists q_1,q_2,q_3, (q_1,(q_2,q_3)), 
(q_s,\Id,q_1),(q_2,\Id,q'_s),(q_3,\Id,q''_s) \in \Delta\}\\
& \cup \{((q_s,q'_s),q''_s) \mid \exists q_1,q_2,q_3, ((q_1,q_2),q_3), 
(q_s,\Id,q_1),(q'_s,\Id,q_2),(q_3,\Id,q''_s) \in \Delta\}
\end{align*}

\begin{lemma}
 $A_4$ and $A_5$ recognise the same relation.
\end{lemma}

\begin{proof}
We prove it by an induction on the structure of relations similar to the one of 
step 1, so we leave it to the reader.
 \qed
\end{proof}

%
%
%

\paragraph*{Step 6:}
We now split the control states set into two parts:
\begin{itemize}
 \item $Q_T$, the states which are target of all and only test transitions and 
source of no test transition,
 \item $Q_C$, the states which are source of all test transitions and target of 
no test transition.
\end{itemize}
Given automaton $A_5 = (Q,I,F,\Delta)$ obtained from the previous step, we 
define $A_6 = (Q',I',F',\Delta')$ with
\begin{align*}
  Q' & = \{q_T,q_C\mid q\in Q\},
  \\
  I' & = \{q_C \mid q\in I\},
  \\
  F' & = \{q_T,q_C \mid q \in F\},
  \\
  \Delta' & = \{(q_C,\theta,q'_C),(q_T,\theta,q'_C) \mid
  (q,\theta,q')\in \Delta, \theta \in \Ops{n-1} \cup \Tests{n-1}
  \{\bcopy{n}{1},\nbcopy{n}{1}\}\}
  \\
  & \cup \{((q_C,q'_C),q''_C),((q_C,q'_T),q''_C),((q_T,q'_C),q''_C),
  ((q_T,q'_T),q''_C) \mid ((q,q'),q'')\in \Delta\}
  \\
  & \cup \{(q_C,(q'_C,q''_C)), (q_T,(q'_C,q''_C)) \mid (q,(q',q'')) \in 
\Delta\} \\
  & \cup \{(q_C,T_L,q'_T) \mid (q,T_L,q')\in \Delta\}.
\end{align*}

\begin{lemma}
 $A_5$ and $A_6$ recognise the same relation.
\end{lemma}

\begin{proof}
As, from step 4 it is not possible to have two successive test transitions, the 
set of recognised operations is the same in both automata, only the labelling 
is modified. The details are left to the reader.
 \qed
\end{proof}
\medskip

Finally, we suppose that an automaton obtained by these steps is distinguished, 
i.e. initial states are target of no transition and final states are source of 
no transition. If not, we can distinguish it by a classical transformation (as 
in the case of word automata).
We now have a normalised automaton with tests $A_6$ obtained after the 
application of the six steps which recognises the same relation as the initial 
automaton $A$.
In subsequent constructions, we will be considering the subsets of
states $Q_T,Q_C,Q_d,Q_c$ as defined in steps 6 and 3, and $Q_{u,d} =
Q_u \cap Q_d$ with $u\in \{T,C\}$ and $d\in \{d,c\}$.
 \qed
\end{proof}

\section{Finite set interpretation}\label{annex:fsi}

In this section, we formally define a finite set interpretation $I_R$
from $\Delta_{\Sigma\cup \{1,2\}}^n$ to the rewriting graph of a
GSTRS $R$. In the whole section,
we consider a distinguished normalised automaton with tests $A = (Q,I,F,\Delta)$
recognising $R^*$, constructed according to the process of the previous
section. 

Let us first formally define a possible presentation of the graph
$\Delta_{\Sigma\cup \{1,2\}}^n$. Vertices of this graph are $n$-stacks
over alphabet $\Sigma \cup \{1,2\}$, and there is an edge
$(x,\theta,y)$ in $\Delta_{\Sigma\cup \{1,2\}}^n$ if $\theta \in
\Ops{n}(\Sigma \cup \{1,2\}) \cup \Tests{n}$ and $y = \theta(x)$.

Since we are building an unlabelled graph, our interpretation
consists of these formul\ae{}:
\begin{itemize}
\item $\delta(X)$ which describes which subsets of
  $\Stacks_n(\Sigma\cup\{1,2\})$ are in the graph,
  \item $\Psi_D(X_s,X_t)$ which is true if $\mathcal{R}_D(s,t)$, for 
$D \in R$,
\item $\phi(X_s,X_t)$ which is true if $\mathcal{R}(A)(s,t)$.
\end{itemize}

\subsection{Notations and Technical Formul\ae{}}

We will use the $\push{d}$ and $\pop{d}$ operations to simplify the notations. 
They have the usual definition (as can be encountered in \cite{Carayol05}), but 
notice that we can define them easily with our operations:
$\push{d}(x) = y$ if there exists $z\in V, a\in \Sigma \cup \{1,2\}$ such that 
$ x \xrightarrow{\cop{1}} z \xrightarrow{\rew{a}{d}} y$, and $\pop{d}(x) = y$ 
if $x = \push{d}(y)$. Observe that $\push{d}(x)$ and $\pop{d}(x)$ are well 
defined as there can only be one $a$ such that the definition holds: the $a$ 
which is the topmost letter of $x$.
We extend this notations to push and pop words to simplify notations.

We first define some formul\ae{} over $\Delta_{\Sigma\cup \{1,2\}}^n$ which 
will be used to construct the set of stacks used to represent stack trees over 
$\Delta_{\Sigma\cup \{1,2\}}^n$.

Given $\theta \in \Ops{n-1}(\Sigma)\cup \Tests{n-1}$, we define $\psi_\theta$ 
such that, given 
two $n$-stacks $x,y$, $\psi_\theta(x,y) = x \xrightarrow{\theta} y$.
$\psi_{\bcopy{n}{i},d}(x,y) = \exists a \in 
\Sigma,z_1,z_2,z_3,z_4,z_5,z_6,z_7,z_8 \in V, x \xrightarrow{\cop{1}} z_1 
\xrightarrow{\rew{a}{i}} z_2 \xrightarrow{\cop{1}} z_3 \xrightarrow{\rew{i}{d}} 
z_4 \xrightarrow{\cop{n}} z_5 \xrightarrow{\rew{d}{i}} z_6 
\xrightarrow{\ncop{1}} z_7 \xrightarrow{\rew{i}{a}} z_8 \xrightarrow{\ncop{1}} 
y$.

$\psi_\theta(x,y)$ is true if $y$ is obtained by applying $\theta$ to $x$. 
$\psi_{\bcopy{n}{i},d}(x,y)$ is true if $y$ is obtained by adding $i$ and $d$ 
to the topmost $1$-stack of $x$, duplicating its topmost $(n-1)$-stack and then 
removing $d$ and $i$ from its topmost $1$-stack.

We now give a technical formula which ensures that a given stack $y$ is 
obtained from a stack $x$ using only the previous formul\ae{}: $\Succ(x,y)$
\begin{multline*}
  \Succ(x,y) = \forall X, ((x \in X \land \forall z,z', (z \in X
  \land (\bigvee_{\theta \in \Ops{n-1} \cup \Tests{n-1}} \psi_\theta(z,z')
  \\
  \vee \bigvee_{i \in \{1,2\}} \bigvee_{d \leq i} 
\psi_{\bcopy{n}{i},d}(z,z')))
  \Rightarrow z' \in X) \Rightarrow y \in X)
\end{multline*}

This formula is true if for every set of $n$-stacks $X$, if $x$ is in $X$ and 
$X$ is closed by the relations defined $\psi_\theta$ and 
$\psi_{\bcopy{n}{i},d}$, then $y$ is in $X$.

\begin{lemma}\label{lem:reach}
  For all $n$-stacks $x = \stack{n}{x_1,\cdots,x_m}$ and $y = 
\stack{n}{y_1,\cdots,y_{m'}}$,

\hfill  $\Succ(x,y)$ holds if and only if $y = [ x_1,\cdots,x_{m-1},
    \push{i_m d_m}(y_m),\push{i_{m+1} d_{m+1}}(y_{m+1}),$ $\cdots,$ $
    \push{i_{m'-1} d_{m'-1}}(y_{m'-1}),y_{m'} ]_n$ where for all $m \leq j < 
m'$, 
$i_j \in  \{1,2\}$, $d_j \leq i_j$ and for all $m \leq j \leq m'$,  there 
exists a sequence of operations $\rho_j \in (\Ops{n-1}(\Sigma)\cup 
\Tests{n-1})^*$ such that
  $\rho_j(x_m,y_j)$.
\end{lemma}


\begin{corollary}\label{cor:code_reach}
  For every $n$-stack $x$ and $a\in \Sigma$,
  $\Succ(\stack{n}{a},x)$ holds if and only if there exist a stack
  tree $t$ and a node $u$ such that $x = \Code{t}{u}$.
\end{corollary}

\begin{proof}
Suppose that there exist a stack tree $t$ and a node $u$ such that $x = 
\Code{t}{u}$. Then $x = [ \push{\#(\varepsilon) u_1}(t(\varepsilon)), 
\push{\#(u_{\leq 1}) u_2}(t(u_{\leq 1})),$ $\cdots,$

\hfill $\push{\#(u_{\leq |u|-1}) 
u_{|u|}} $ $(t(u_{\leq |u|-1})), t(u)]_n$.
As for every $i$, $t(u_{\leq i})$ is in $\Stacks_{n-1}(\Sigma)$, there 
exists a $\rho_i$ in $(\Ops{n-1}(\Sigma)\cup \Tests{n-1})^*$ such that 
$\rho_i(\stack{n}{a},t(u_{\leq i}))$. Then by the previous lemma, 
$\Succ(\stack{n}{a},x)$ is true.

Conversely, suppose that $\Succ(\stack{n}{a},x)$ is true. By Lemma
\ref{lem:reach}, we therefore have $x = \stack{n}{\push{i_0 d_0}(x_0),
  \push{i_1 d_1}(x_1), \cdots, \push{i_{m-1} d_{m-1}}(x_{m-1}), x_m}$, where for
every $j$ there exists a $\rho_j \in (\Ops{n-1}(\Sigma)\cup \Tests{n-1})^*$ 
such that
$x_j = \rho_j(\stack{n}{a})$. Then, for every $j$, $x_j \in
\Stacks_{n-1}(\Sigma)$.

We take a tree domain $U$ such that $d_0\cdots d_{m-1} \in U$. We define a 
tree $t$ of domain $U$ such that for every $j$, $t(d_0\cdots d_j) = x_{j+1}$,
$t(\varepsilon) = x_0$, every node $d_0 \cdots d_j$ has $i_{j+1}$ sons, the 
node $\varepsilon$ has $i_0$ sons, and for every $u \in U$ which is not a $d_0 
\cdots d_j$, $t(u) = \stack{n}{a}$. Then we have $x = \Code{t}{d_0\cdots 
d_{m-1}}$.
 \qed
\end{proof}

\subsection{The formula $\delta$}

\noindent
We now define $\delta(X) = \mathrm{OnlyLeaves}(X)) \land
\mathrm{TreeDom}(X) \land \mathrm{UniqueLabel}(X)$ with
\begin{alignat*}{2}
& \mathrm{OnlyLeaves}(X)&\ =\ & \forall x, x\in X \Rightarrow 
\Succ(\stack{n}{a},x)\\
& \mathrm{TreeDom}(X) &\ =\ & \forall x,y,z ((x \in X \land 
 \psi_{\bcopy{n}{2},2}(y,z) \land \Succ(z,x)) \Rightarrow\\
& &	    &	      \exists r,z' (r \in X \land 
\psi_{\bcopy{n}{2},1}(y,z') \land \Succ(z',r))) \wedge\\
& & & ((x \in X \land  \psi_{\bcopy{n}{2},1}(y,z) \land \Succ(z,x)) 
\Rightarrow\\
& &	    &	      \exists r,z' (r \in X \land 
\psi_{\bcopy{n}{2},2}(y,z') \land \Succ(z',r)))\\
& \mathrm{UniqueLabel}(X)&\ =\ & \forall x,y, (x \neq y \land x \in X \land 
y \in X) 
\Rightarrow\\
& &	    &	      (\exists z,z',z'', \psi_{\bcopy{n}{2},1}(z,z') \land 
\psi_{\bcopy{n}{2},2}(z,z'') \land\\
& &	    &	      ((\Succ(z',x) \land \Succ(z'',y)) \vee (\Succ(z'',x) 
\land \Succ(z',y))))
\end{alignat*}
where $a$ is a fixed letter of $\Sigma$.

Formula $\mathrm{OnlyLeaves}$ ensures that an element $x$ in $X$ encodes a
node in some stack tree.
$\mathrm{TreeDom}$ ensures that the prefix closure of
the set of words $d_0 \cdots d_{m-1}$ such that 
\[
\stack{n}{\push{i_0 d_0}(x_0)),\push{i_1 
d_1}(x_1),\cdots,\push{i_{m-1} d_{m-1}}(x_{m-1}),x_m }
\in X
\]
is a valid domain of a tree, and that the set of words $i_0\cdots i_{m-1}$ is 
included in this set (in other words, that the arity announced by the $i_j$ is 
respected).
an 
Finally $\mathrm{UniqueLabel}$ ensures that for any two elements 
\begin{align*}
  & x = \stack{n}{\push{i_0 d_0}(x_0)), \push{i_1 d_1}(x_1), \cdots,
    \push{i_{m-1} d_{m-1}}(x_{m-1}), x_m}
  \\
  \text{and } & y = \stack{n}{\push{i'_0 d'_0}(y_0)), \push{i'_1 d'_1}(y_1),
    \cdots, \push{i'_{m-1} d'_{m'-1}}(y_{m'-1}), y_{m'} }
\end{align*}
of $X$, there exists an index $1\leq j\leq \mathrm{min}(m,m')$ such that for
every $k< j$, $x_k = y_k$, $i_k = i'_k$ and $d_k = d'_k$, $x_j = y_j$, $i_j = 
i'_j$ and $d_j \neq d'_j$, i.e. for any two elements, the $(n-1)$-stacks 
labelling common ancestors are equal, and $x$ and $y$ cannot encode the same 
leaf (as $d_0\cdots d_{m-1} \neq d'_0 \cdots d'_{m'-1}$). Moreover, it also 
prevents $x$ to code a node on the path from the root to the node coded by $y$.

\begin{lemma}
$\forall X \subseteq \Stacks_n(\Sigma \cup \{1,2\})$, $\delta(X) \iff \exists 
t \in \TS_n, X = X_t$

where $X$ ranges only over finite sets of $\Stacks_n(\Sigma \cup \{1,2\})$.
\end{lemma}

\begin{proof}
 We first show that for every $n$-stack tree $t$, $\delta(X_t)$ holds over 
$\Delta_{\Sigma\cup\{1,2\}}^n$.
 By definition, for every $x \in X_t$, $\exists u \in fr(t), x=\Code{t}{u}$, 
and 
then $\Succ(\stack{n}{a},x)$ holds (by Corollary \ref{cor:code_reach}).
Thus $\mathrm{OnlyLeaves}$ holds.
 
 Let us take $x \in X_t$ such that $x = \Code{t}{u}$ with $u = u_0\cdots u_i 2 
u_{i+2} \cdots u_{|u|}$. As $t$ is a tree, $u_0 \cdots u_i 2 \in dom(t)$ and 
so is $u_0 \cdots u_i 1$. Then, there exists $v \in fr(t)$ such 
that $\forall j \leq i, v_j = u_j$, $v_{i+1} = 1$, and $\Code{t}{v} \in X_t$.
Let us now take $x \in X_t$ such that $x = \Code{t}{u}$ with $u = u_0 \cdots 
u_i 1 u_{i+2} \cdots u_{|u|}$ and $\#(u_0 \cdots u_i 1) = 2$, then $u_0 
\cdots u_i 2$ is in $dom(t)$ and there exists $v \in fr(t)$ such that $\forall 
j 
\leq i, v_j = u_j$, $v_{i+1} = 2$ and $\Code{t}{v} \in X_t$.
Thus $\mathrm{TreeDom}$ holds.
 
 Let $x$ and $y$ in $X_t$ such that $x \neq y$, $x=\Code{t}{u}$ and $y 
= \Code{t}{v}$, and let $i$ be the smallest index such that $u_i \neq v_i$. 
Suppose that $u_i = 1$ and $v_i = 2$ (the other case is symmetric). We 
call $z = \Code{t}{u_0\cdots u_{i-1}}$, and take $z'$ and $z''$ such that 
$\psi_{\bcopy{n}{2},1}(z,z')$ and $\psi_{\bcopy{n}{2},2}(z,z'')$. We have then 
$\Succ(z',x)$ and $\Succ(z'',y)$.
And thus $\mathrm{UniqueLabel}$ holds. Therefore, for every stack tree
$t$, $\delta(X_t)$ holds.  

 Let us now show that for every $X \subseteq \Stacks_n(\Sigma \cup \{1,2\})$ 
such that $\delta(X)$ holds, there exists $t \in \TS_n$, such that $X = X_t$.
 As $\mathrm{OnlyLeaves}$ holds, for every $x\in X$, 
\[
x = 
\stack{n-1}{\push{i_0 u_0}(x_0),\push{i_1 u_1}(x_1),\cdots,\push{i_{k-1} u_{k-1 
} } (x_ { k-1 } ) , x_k}
\]
with, for all $j$, $x_j \in \Stacks_{n-1}$, $i_j \in \{1,2\}$ and $u_j \leq 
i_j$. In the following, we denote by $u^x$ the word $u_0 \cdots u_{k-1}$ for a
given $x$, and by $U = \{u \mid \exists x \in X, u \sqsubseteq
u^x\}$. $U$ is closed under prefixes. 
 As $\mathrm{TreeDom}$ holds, for all $u$, if $u2$ is in $U$, then $u1$ 
is in $U$ as well. Therefore $U$ is the domain of a tree. Moreover, if there is 
a $x$ such that $u1 \sqsubseteq u^x$ and $i_{|u|} = 2$, then 
$\mathrm{TreeDom}$ ensures that there is $y$ such that $u2 \sqsubseteq u^y$ and 
thus $u2 \in U$.
 As $\mathrm{UniqueLabel}$ holds, for every $x$ and $y$ two distinct elements 
of $X$, there exists $j$ such that for all $k < j$ we have $u^x_k= u^y_k$, 
and $u^x_j \neq u^y_j$. Then, for all $k \leq j$, we have $x_k = y_k$ and $i_k 
= i'_k$. Thus, for every $u \in U$, we can define $\sigma_u$ such that for 
every $x$ such that $u \sqsubseteq u^x$, $x_{|u|} = \sigma_u$, and the number 
of sons of each node is consistent with the coding.
 \smallskip
 
 Consider the tree $t$ of domain $U$ such that for all $u \in U$,
 $t(u) = \sigma_u$. We have $X = X_t$, which concludes the proof.
 \qed
\end{proof}

\subsection{The formula $\Psi_D$ associated with an operation}

We now take an operation $D$ which we suppose to be reduced, for the sake 
of simplicity (but we could do so for a non reduced operation, and for any 
operation, there exists a reduced operation with tests defining the same 
relation, from the two previous appendices). We define 
inductively $\psi_D$ as follow:

\begin{itemize}
 \item $\Psi_\emptydag (X,Y) = (X = Y)$
 
 \item $\Psi_{(F \cdot_{1,1} D_\theta) \cdot_{1,1} G}(X,Y) = \exists 
,z,z',Z,X',Y', z \in Z \wedge X\backslash X' = Y 
\backslash Y' = Z \backslash\{z\} \wedge \psi_\theta(z,z') \wedge 
\Psi_F(X,Z) \wedge \Psi_G(Z \cup \{z'\}\backslash\{z\},Y)$, for 
$\theta 
\in \Ops{n-1}\cup \Tests{n}$
 
 \item $\Psi_{(F \cdot_{1,1} D_{\bcopy{n}{1}}) \cdot_{1,1} G}(X,Y) = 
\exists z,z',Z,X',Y', z \in Z \wedge 
X\backslash X' = Y \backslash Y' = Z \backslash\{z\} \wedge 
\psi_{\bcopy{n}{1},1}(z,z') \wedge \Psi_F(X,Z) \wedge \Psi_G(Z \cup 
\{z'\}\backslash\{z\},Y)$
 
 \item $\Psi_{(F \cdot_{1,1} D_{\nbcopy{n}{1}}) \cdot_{1,1} G}(X,Y) = 
\exists z,z',Z,X',Y', z \in Z \wedge X\backslash 
X' = Y \backslash Y' = Z \backslash\{z\} \wedge \psi_{\bcopy{n}{1},1}(z',z) 
\wedge \Psi_F(X,Z) \wedge \Psi_G(Z \cup \{z'\}\backslash\{z\},Y)$
 
 \item $\Psi_{((F \cdot_{1,1} D_{\bcopy{n}{2}}) \cdot_{1,2} H) 
\cdot_{1,1} G}(X,Y) = \exists z,z',z'',Z,Z',X',Y', z \in Z \wedge 
X\backslash X' = Y \backslash Y' = Z \backslash\{z\} \wedge 
\psi_{\bcopy{n}{1},2}(z,z') \wedge \psi_{\bcopy{n}{2},2}(z,z'') \wedge 
\Psi_F(X,Z) \wedge \Psi_G(Z \cup
\{z',z''\}\backslash$ $\{z\},Z') \wedge z'' \in Z' \wedge z' \notin Z' 
\wedge 
\Psi_H(Z',Y)$
 
 \item $\Psi_{(F \cdot_{1,1} (G \cdot_{1,2} D_{\nbcopy{n}{1}})) 
\cdot_{1,1} H}(X,Y) = 
\exists z,z',z'',Z,Z',X',Y', z \in Z \wedge z' \in Z \wedge z \in Z' \wedge z' 
\notin Z' \wedge X\backslash X' = Y 
\backslash Y' = Z \backslash\{z,z'\} \wedge \psi_{\bcopy{n}{2},1}(z'',z) 
\wedge \psi_{\bcopy{n}{2},2}(z'',z') \wedge \Psi_F(X,Z') \wedge 
\Psi_G(Z',Z) \wedge  \Psi_G(Z \cup \{z''\}\backslash\{z,z'\},Y)$
\end{itemize}

As $D$ is a finite DAG, every $\psi_D$ is a finite formula, and is 
thus a monadic formula.

This formula is true if its two arguments are related by $\mathcal{R}_D$.

\begin{proposition}
 Given two stack trees $s,t$ and an operation $D$, $t\in D(t)$ if and 
only if $\Psi_D(X_s,X_t)$ is true.
\end{proposition}

\begin{proof}
 We show it by induction on the structure of $D$:
 \begin{itemize}
  \item If $D = \emptydag$, $\Psi_D(X_s,X_t)$ if and only if $X_s = 
X_t$, which is true if and only if $s = t$.

  \item $D = (F \cdot_{1,1} D_\theta) \cdot_{1,1} G$, with 
$\theta \in \Ops{n-1}\cup \Tests{n}$. Suppose $t \in D(s)$, there exists 
$i$ such that $t = D_{(i)}(t)$. By definition, $t = G_{(i)} 
(\theta_{(i)} (F_{(i)}(s)))$. We call $r = F_{(i)} (s)$. By induction 
hypothesis, we have $\Psi_F(X_s,X_r)$. By definition, we have, for all 
$j<i$, $\Code{s}{u_j} = \Code{r}{u_j}$, and for all $j>i$, 
$\Code{s}{u_{j+|I_F|-1}} = \Code{r}{u_j}$, thus $X_s \backslash 
\{\Code{s}{u_j} \mid i \leq j \leq |I_F|-1\} = X_r \backslash 
\{\Code{r}{u_i}\}$. We call $r' = \theta_{(i)}(r)$. We have $X_{r'} = X_r 
\backslash \{\Code{r}{u_i}\} \cup \{\theta(\Code{r}{u_i})\}$. And by 
definition, 
we have $\psi_\theta(\Code{r}{u_i},\theta(\Code{r}{u_i}))$. We have $t= 
G_{(i)}(r')$, thus, by induction hypothesis, $\Psi_G(X_{r'},X_t)$ is 
true. Moreover, by definition, $X_t \backslash \{\Code{t}{u_j} \mid i \leq j 
\leq |O_G|-1\} = X_{r'} \backslash $ $\{\Code{r'}{u_i}\} = X_r \backslash 
\{\Code{r}{u_i}\}$. Thus, $\Psi_D(X_s,X_t)$ is true, with $Z = X_r$, $z = 
\Code{r}{u_i}$, $z' = \Code{r'}{u_i}$, $X' = \{\Code{s}{u_j} \mid i \leq j \leq 
|I_D|-1\}$ and $Y' = \{\Code{t}{u_j} \mid i \leq j \leq |O_D|-1\}$.

Suppose that $\Psi_D(X_s,X_t)$ is true. We call $r$ the tree such that $X_r 
= Z$. By induction hypothesis, we have $r \in F(s)$. Moreover, we have $z = 
\Code{r}{u_i}$ such that $X_r\backslash\{z\} = X_s \backslash X'$. Thus, by 
definition, $r = F_{(i)}(s)$, and $X' = \{\Code{s}{u_j}\mid i\leq 
|I_F|-1\}$. We have $z' = \theta(z)$, as $\psi_\theta(z,z')$ is true. We 
call $r' = \theta_{(i)}(r)$, and we have $X_{r'} = X_r \backslash \{z\} \cup 
\{z'\}$. As we have $\Psi_G(X_{r'},Y)$, by induction, we have $t \in 
G(r')$. As we moreover have $Y \backslash Y' = Z \backslash \{z\}$, we 
thus have $t = G_{(i)}(r')$. Thus, we have $t = 
G_{(i)}(\theta_{(i)}(F_{(i)}(s))) = D_{(i)}(s)$.
 \end{itemize}
 
 The other cases are similar and left to the reader.

\end{proof}

\subsection{The formula $\phi$ associated with an automaton}

Let us now explain $\phi(X,Y)$, which can be written as $\exists
Z_{q_1},\cdots,Z_{q_{|Q|}}, \phi'(X,Y,\vec{Z})$ with
$\phi'(X,Y,\vec{Z}) = \init(X,Y,\vec{Z}) \land \diff(\vec{Z}) \land
\mathrm{Trans}(\vec{Z})$. We detail each of the three subformulas
$\init$, $\diff$ and $\mathrm{Trans}$ below:

\[
\init(X,Y,\vec{Z})= (\bigcup_{q_i \in I} Z_{q_i}) \subseteq X \land 
(\bigcup_{q_i \in F} Z_{q_i}) \subseteq Y \land X \setminus(\bigcup_{q_i \in 
I} Z_{q_i}) = Y \setminus (\bigcup_{q_i \in F} Z_{q_i})
\]
This formula is here to ensure that only leaves of $X$ are labelled by initial 
states, only leaves of $Y$ are labelled by final states and outside of their 
labelled leaves, $X$ and $Y$ are equal (i.e. not modified).

\begin{multline*}
  \diff(\vec{Z}) = \big( \bigwedge_{q,q' \in Q_{T,c}} Z_q \cap Z_{q'}
  = \emptyset \big) \land \big(\bigwedge_{q,q' \in Q_{C,c}} Z_q \cap
  Z_{q'} = \emptyset \big)\\
  \land \big( \bigwedge_{q,q' \in Q_{T,d}} Z_q \cap Z_{q'} = \emptyset
  \big) \land \big( \bigwedge_{q,q' \in Q_{C,d}} Z_q \cap Z_{q'} =
  \emptyset \big)
\end{multline*}
This formula is here to ensure that a given stack (and thus a given leaf in a 
tree of the run) is labelled by at most a state of each subpart of $Q$: 
$Q_{T,d},Q_{C,d},Q_{T,c},Q_{C,c}$. So if we have a non deterministic choice to 
do we will only choose one possibility. 

\[
\mathrm{Trans}(\vec{Z}) = \forall s, \bigwedge_{q\in Q} ((s\in Z_q)
\Rightarrow (\bigvee_{K \in \Delta}
\mathrm{Trans}_{K}(s,\vec{Z}) \vee \rho_q))
\]
where $\rho_q$ is true if and only if $q$ is a final state, and 
\begin{alignat*}{2}
  & \mathrm{Trans}_{(q,\bcopy{n}{1},q')}(s,\vec{Z}) & & = \exists t,
  \psi_{\bcopy{n}{1},1}(s,t) \land t \in Z_{q'},\\
  & \mathrm{Trans}_{(q,\nbcopy{n}{1},q')}(s,\vec{Z}) & & = \exists t,
  \psi_{\bcopy{n}{1},1}(t,s) \land t \in Z_{q'},\\
  & \mathrm{Trans}_{(q,\theta,q')}(s,\vec{Z}) & & = \exists t,
  \psi_{\theta}(s,t) \land t \in Z_{q'}, \text{for } \theta \in \Ops{n-1}\cup 
\Tests{n-1},\\
  & \mathrm{Trans}_{(q,(q',q''))}(s,\vec{Z}) & & = \exists t,t',
  \psi_{\bcopy{n}{2},1}(s,t) \land \psi_{\bcopy{n}{2},2}(s,t') \land t
  \in Z_{q'} \land t' \in Z_{q''},\\
  & \mathrm{Trans}_{((q,q'),q'')}(s,\vec{Z}) & & = \exists t,t',
  \psi_{\bcopy{n}{2},1}(t',s) \land \psi_{\bcopy{n}{2},2}(t',t) \land t
  \in Z_{q'} \land t' \in Z_{q''},\\
  & \mathrm{Trans}_{((q',q),q'')}(s,\vec{Z}) & & = \exists t,t',
  \psi_{\bcopy{n}{2},1}(t',t) \land \psi_{\bcopy{n}{2},2}(t',s) \land t
  \in Z_{q'} \land t' \in Z_{q''}.
\end{alignat*}
This formula ensures that the labelling respects the rules of the automaton, 
and that for every stack labelled by $q$, if there is a rule starting by $q$, 
there is at least a stack which is the result of the stack by one of those 
rules. And also that it is possible for a final state to have no successor.

\begin{proposition}
 Given $s,t$ two stack trees, 
$\phi(s,t)$ if and only if there are some operations $D_1,\cdots,D_k$ 
recognised by $A$ such that $t$ is obtained by applying 
$D_1,\cdots,D_k$ at disjoint positions of $s$.
\end{proposition}

\begin{proof}
  First suppose there exist such $D_1,\cdots,D_k$.
We construct a labelling of $\Stacks_n(\Sigma 
\cup \{1,2\})$ which satisfies $\phi(X_s,X_t)$.
 We take a labelling of the $D_i$ by $A$. We will label the $\Stacks_n$ 
according to this labelling. If we obtain a tree $t'$ at any step in the 
run of the application of $D_i$ to $s$, we label $\Code{t'}{u}$ by 
the labelling of the node of $D_i$ appended to the leaf at position $u$ of 
$t'$. Notice that this does not depend on the order we apply the $D_i$ to 
$s$ nor the order of the leaves we choose to apply the operations first.

We suppose that $t = {D_k}_{i_k} (\cdots {D_1}_{i_1}(s)\cdots)$. 
Given a node $x$ of an $D_i$, we call $l(x)$ its labelling.

Formally, we define the labelling inductively: the 
$(D_1,i_1,s_1),\cdots,(D_k,i_k,s_k)$ labelling of 
$\Stacks_n(\Sigma\cup \{1,2\})$ is the following.
\begin{itemize}
  \item The $\emptyset$ labelling is the empty labelling.
  \item The $(D_1,i_1,s_1),\cdots,(D_k,i_k,s_k)$ labelling is the 
union of the $(D_1,i_1,s_1)$ labelling and the 
$(D_2,i_2,s_2),\cdots,(D_k,i_k,s_k)$ labelling.
 \item The $(\emptydag,i,s)$ labelling is $\{\Code{s}{u_i} \rightarrow l(x)\}$, 
where $u_i$ is the $i^{th}$ leaf of $s$ and $x$ is the unique 
node of $\emptydag$.
  \item The $(F_1 \cdot_{1,1} D_\theta) \cdot_{1,1} F_2,i,s)$ labelling 
is the $(F_1,i,s),(F_2,i,\theta_{(i)} ({F_1}_{(i)}(s)))$ labelling.

  \item The $((((F_1 \cdot_{1,1} D_{\bcopy{n}{2}}) \cdot_{2,1} F_3) 
\cdot_{1,1} F_2),i,s)$ labelling is the 
$(F_1,i,s),$ 

\hfill $(F_2,i,{\bcopy{n}{2}}_{(i)}($ ${F_1}_{(i)}(s)))$ 
$,(F_3,i+1
,$ ${ \bcopy{n}{2}}_{(i)}({F_1}_{(i)}(s)))$ labelling.
%

 \item The $((F_1 \cdot_{1,1} (F_2 \cdot_{2,1} \nbcopy{n}{2})) 
\cdot_{1,1} F_3,i,s)$ labelling is the 
$(F_1,i,s),(F_2,i+|I_{F_1}|,s),$ $(F_3,i,{\nbcopy{n}{2}}_{(i)}
({F_2}_{(i+1)} ({F_1}_{(i)}(s))))$ labelling.
\end{itemize}


Observe that this process terminates, as the sum of the edges and the nodes of 
all the DAGs strictly diminishes at every step.

We take $\vec{Z}$ the  
$(D_1,i_1,s),\cdots,(D_k,i_k,s)$ labelling of 
$\Stacks_{n}(\Sigma\cup\{1,2\})$.
 
 \begin{lemma}
  The labelling previously defined $\vec{Z}$ satisfies $\phi'(X_s,X_t,\vec{Z})$.
 \end{lemma}
 
 \begin{proof}
   
   Let us first cite a technical lemma which comes directly from the definition 
of the labelling:
    \begin{lemma}\label{lemma:labelling}
     Given a reduced operation $D$, a labelling of $D$, 
$\rho_D$, a stack tree $t$, a $i\in \mathbb{N}$ and a $j\leq 
|I_D|$, the label of $\Code{t}{u_{i+j-1}}$ (where $u_i$ is the 
$i^{\text{th}}$ leaf of $t$) in the $(D,i,t)$ labelling is 
$\rho_D(x_j)$ (where $x_j$ is the $j^{\text{th}}$ input node of $D$).
    \end{lemma}
    
    For the sake of simplicity, let us consider for this proof that $D$ is 
a reduced operation (if it is a set of reduced operations, the proof is the 
same for every operations).

   First, let us prove that $\mathrm{Init}$ is satisfied.
   From the previous lemma, all nodes of $X_s$ are labelled with 
the labels of input nodes of $D$ (or not labelled), thus they are 
labelled by initial states (as we considered an accepting labelling of 
$D$). Furthermore, as the automaton is distinguished, only these one 
can be labelled by initial states.
   Similarly, the nodes of $X_t$, and only them are labelled by final states 
(or not labelled).
   

    We now show that $\mathrm{Trans}$ is satisfied. Let us suppose that a 
$\Code{t'}{u_i}$ is labelled by a $q$. By construction of the labelling, it has 
been obtained by a $(\emptydag, i,t')$ labelling. If $q$ is final, then we have 
nothing to verify, as $\rho_q$ is true. If not, the node $x$ labelled by $q$ 
which is the unique node of the $\emptydag$ which labelled $\Code{t'}{u_i}$ by 
$q$ has at least one son in $D$. Suppose, for instance that $D = 
(F_1 \cdot_{1,1} D_\theta) \cdot_{1,1} F_2$ such that $x$ is the output node 
of $F_1$. We call $y$ the input node of $F_2$. As $D$ is 
recognised by $A$, it is labelled by a $q'$ such that 
$(q,\theta,q')\in\Delta_A$. By construction, we take the 
$(F_1,i,s),(F_2,i,\theta_{(i)}(t'))$ labelling, with 
$t'={F_2}_{(i)}(s)$. Thus we have $\Code{\theta_{(i)}(t')}{u_i}$ labelled 
by $q'$ (from Lemma \ref{lemma:labelling}), and thus 
$\mathrm{Trans}_{(q,\theta,q')}(\Code{t'}{u_i},\vec{Z})$ is true, as 
$\psi_\theta(\Code{t'}{u_i},\Code{\theta_{(i)}(t')}{u_i}$ is true.

The other possible cases for decomposing $D$ ($D = (((F_1 
\cdot_{1,1} D_{\bcopy{n}{1}}) \cdot_{2,1} F_3) \cdot_{1,1} F_2$ or 
$D = ((F_1 \cdot_{1,1} (F_2 \cdot_{2,1} \nbcopy{n}{2})) 
\cdot_{1,1} F_3$) are 
very similar and are thus left to the reader. Observe that $D$ may not be 
decomposable at the node $x$, in which case we decompose $D$ and consider 
the part containing $x$ until we can decompose the DAG at $x$, where the 
argument is the same.

Let us now prove that the labelling satisfies $\mathrm{Diff}$.
Given $q,q' \in Q_{C,d}$, suppose that there is a $\Code{t'}{u_i}$ which is 
labelled by $q$ and $q'$. By construction, this labelling is obtained by a 
$(F_1,i,t'_1),(F_2,i,t'_2)$ labelling, where $F_1$ and $F_2$ 
are both $\emptydag$, and $t'_1(u_i) = t'_1(u_i)$. We call $x$ (resp. $y$) 
the unique node of $F_1$ (resp. $F_2$). $x$ is labelled by $q$ and $y$ 
by $q'$.

Suppose that $D$ can be decomposed as $(G \cdot_{1,1} D_\theta )
\cdot_{1,1} H$ (or $((G \cdot_{1,1} D_{\bcopy{n}{2}}) \cdot_{2,1} 
K) \cdot_{1,1} H$, or $((G \cdot_{1,1} (H \cdot_{1,2} 
D_{\nbcopy{n}{2}}) \cdot_{1,1} K$)  such 
that $y$ is the output node of $G$ (if not, decompose $D$ until you 
can obtain such a decomposition). Then, suppose you can decompose $G = 
G_1 \cdot_{1,1} D_\theta \cdot_{1,1} G_2$ (or $((G_1 
\cdot_{1,1} (G_3 \cdot_{1,2} D_{\nbcopy{n}{2}}) \cdot_{1,1} G_2$. As 
we are considering states of $Q_{C,d}$, there is no other possible case) such 
that $x$ is the input node of $G_2$. Thus, we have by construction 
$G_2 (\Code{t'}{u_i}) = \Code{t'}{u_i}$. So $G_2$ defines a relation 
contained in the identity. As it is a part of $D$ and thus labelled by 
states of $A$, with $q$ and $q'$ in $Q_{C,d}$, there is no $\bcopy{n}{j}$ nor 
$\nbcopy{n}{j}$ transitions in $G_2$. Moreover, as $q$ and $q'$ are in 
$Q_{C,d}$, $G_2$ is not a single test transition. Then it is a sequence of 
elements of $\Ops{n-1}\cup \Tests{n-1}$ defining a relation included into the 
identity. As $A$ 
is normalised, this is impossible, and then $\Code{t'}{u_i}$ cannot be labelled 
by both $q$ and $q'$.

Taking two states in the other subsets of $Q$ yields the same contradiction 
with few modifications and are thus left to the reader.

 
 Then, as all its sub-formul\ae{} are true, $\phi'(X_s,X_t,\vec{Z})$ is true 
with the described labelling $\vec{Z}$. And then $\phi(X_s,X_t)$ is true.
 \qed
 \end{proof}
 
 Suppose now that $\phi(X_s,X_t)$ is satisfied. We take a minimal labelling 
$\vec{Z}$ that satisfies the formula $\phi'(X_s,X_t,\vec{Z})$.
 We construct the following graph $D$ :
\[
\begin{array}{lll}
     V_D & = & \{(x,q) \mid x \in \Stacks_n(\Sigma \cup \{1,2\}) \land x \in 
     Z_q\}
     \\
     E_D & =    & \{((x,q),\theta,(y,q')) \mid (\exists \theta, 
(q,\theta,q') \in \Delta \land \psi_\theta(x,y))\}\\
     & \cup & \{((x,q),1,(y,q')), 
     ((x,q),2,(z,q'')) \mid (q,(q',q'')) \in \Delta\\
     &        & \land \psi_{\bcopy{n}{2},1}(x,y) \land 
     \psi_{\bcopy{n}{2},2}(x,z)\}\\
     & \cup & \{((x,q),\bar{1},(z,q'')), 
     ((y,q'),\bar{2},(z,q'')) \mid ((q,q'),q'') \in \Delta\\
     &	 & \land \psi_{\bcopy{n}{2},1}(z,x) \land 
     \psi_{\bcopy{n}{2},2}(z,y)\}\\
     & \cup & \{((x,q),1,(y,q')) \mid (q,\bcopy{n}{1},q') \in \Delta \land 
\psi_{\bcopy{n}{1},1}(x,y)\}\\
     & \cup & \{((x,q),\bar{1},(y,q')) \mid (q,\nbcopy{n}{1},q') \in \Delta 
\land \psi_{\bcopy{n}{1},1}(y,x) \}
   \end{array}
\]

 \begin{lemma}
  $D$ is a disjoint union of operations 
$D_1,\cdots,D_k$.
 \end{lemma}
 
 \begin{proof}
  Suppose that $D$ is not a DAG, then there exists $(x,q) \in V$ such that 
$(x,q) \xrightarrow{+} (x,q)$, then there exists a sequence of operations in 
$A_d$ (for $A_c$ it is symmetric, and there is no transition from $A_c$ to 
$A_d$, thus a cycle cannot have states of the both parts) which is the identity 
(and thus it is an sequence of operations of $\Ops{n-1}\cup \Tests{n-1}$). As 
$A_d$ is normalised, it is not possible to have such a sequence.
  Then, there is no cycle in $D$ which is therefore a DAG.
  
  By definition of $E_D$, it is labelled by $\Ops{n-1}\cup \Tests{n-1} \cup
\{1,\bar{1},2,\bar{2}\}$.

  We choose an $D_i$. Suppose that it is not an operation. Thus, there 
exists a node $(x,q)$ of $D_i$ such that $D_i$ cannot be decomposed 
at this node (i.e, in the inducted decomposition, there will be no case which 
can be applied to cut either $D_i$ or one of its subDAG to obtain $(x,q)$ 
as the output node of a subDAG obtained (or the input node). Let us consider 
the following cases for the neighbourhood of $(x,q)$:
  
  \begin{itemize}
   \item $(x,q)$ has a unique son $(y,q')$, which has no other father such that 
$(x,q) \xrightarrow{2} (y,q')$. By 
definition of $\mathrm{Trans}$, we have that $\psi_{\bcopy{n}{2},2}(x,y)$, and 
thus we have a $(q,(q'',q'))\in \Delta$ and a $z$ such that 
$\psi_{\bcopy{n}{2},1}(x,z)$ which is in $Z_{q''}$. This contradicts that 
$(x,q)$ has a unique son in $D_i$.
  If $(x,q) \xrightarrow{\bar{2}} (y,q')$, the case is similar.
  For every other $\theta\in \Ops{n-1} \cup \Tests{n-1} \cup \{1,\bar{1}\}$, we 
can decompose 
the subDAG $\{(x,q)\xrightarrow{\theta}(y,q')\}$ as $(\emptydag \cdot_{1,1} 
D_\theta) \cdot_{1,1} \emptydag$.
  
  \item Suppose that $(x,q)$ has at least three sons 
$(y_1,q_1),(y_2,q_2),(y_3,q_3)$. There is no subformula of $\mathrm{Trans}$ 
which impose to label three nodes which can be obtained from $x$, so this 
contradicts the minimality of the labelling.
  
  For a similar reason, $(x,q)$ has at most two fathers.

  \item Suppose that $(x,q)$ has two sons $(y_1,q_1)$ and $(y_2,q_2)$. By 
definition of $\mathrm{Trans}$ and by minimality, we have that 
$\psi_{\bcopy{n}{2},1}(x,y_1)$, $\psi_{\bcopy{n}{2},2}(x,y_2)$, and 
$(q,(q_1,q_2$ $))\in \Delta$ (otherwise, the labelling would not be minimal, as 
it 
is the only subformula imposing to label two sons of a node). Thus we have 
$(x,q) \xrightarrow{1} (y_1,q_1)$ and $(x,q)\xrightarrow{2} (y_2,q_2)$.
  By minimality again, $(y_1,q_1)$ and $(y_2,q_2)$ have no other father than 
$(x,q)$.
  In this case, the subDAG $\{(x,q) \xrightarrow{1} 
(y_1,q_1),(x,q)\xrightarrow{2} (y_2,q_2)\}$ can be decomposed as $((\emptydag 
\cdot_{1,1} D_{\bcopy{n}{2}}) \cdot_{2,1} \emptydag) \cdot_{1,1} \emptydag$.

  \item Suppose that $(x,q)$ has a unique son $(y_1,q_1)$ which has an other 
father $(y_2,q_2)$. By definition of $\mathrm{Trans}$ and by minimality of the 
labelling, we have that $\psi_{\bcopy{n}{2},1}(y_1,x)$, 
$\psi_{\bcopy{n}{2},2}(y_1,y_2)$, and $((q,q_2),q_1)\in \Delta$. Thus we have 
$(x,q) \xrightarrow{\bar{1}} (y_1,q_1)$ and $(y_2,q_2)\xrightarrow{\bar{2}} 
(y_1,q_1)$. By minimality again, $(y_2,q_2)$ has no other son than $(y_1,q_1)$.
  In this case, the subDAG $\{(x,q) \xrightarrow{\bar{1}} (y_1,q_1), 
(y_2,q_2)\xrightarrow{\bar{2}} (y_1,q_1)\}$ can be decomposed as $(\emptydag 
\cdot_{1,1} (\emptydag \cdot_{1,2} D_{\nbcopy{n}{2}})) \cdot_{1,1} \emptydag$.
  \end{itemize}
  
  In all the cases we considered, or the case is impossible, or the DAG is 
decomposable at the node $(x,q)$. Thus, the DAG $D_i$ is always 
decomposable and is thus an operation.
 \qed
 \end{proof}
 
 \begin{lemma}
  Each $D_i$ is recognised by $A$
 \end{lemma}

 \begin{proof}
  By construction, for every node $(x,q)$, if $x \in X_s$, $q$ is an initial 
state (because $\mathrm{init}$ is satisfied), and $(x,q)$ is then an input 
node, as 
$A$ is distinguished. And as $\mathrm{init}$ is satisfied, only these nodes are 
labelled 
by initial states.
  
  Also, for every node $(x,q)$, if $x \in X_t$, $q$ is a final state (because 
$\mathrm{init}$ is satisfied) and $(x,q)$ is then an output node, as $A$ is 
distinguished. And as $\mathrm{init}$ in satisfied, only these nodes are 
labelled by final states.
  
  By construction, the edges are always transitions present in $\Delta$, and 
then we label each node $(x,q)$ by $q$. 
  
  As the formula $\mathrm{Trans}$ is satisfied, we have that given any node 
$(x,q)$, either $q$ is final (and then $(x,q)$ is an output node), or there 
exists one of the following:
  \begin{itemize}
   \item a node $(y,q')$ and $\theta$ such that $\psi_{\theta}(x,y)$ and 
$(q,\theta,q') \in \Delta$
   \item two nodes $(y,q')$ and $(z,q'')$ such that 
$\psi_{\bcopy{n}{2},1}(x,y)$, 
$\psi_{\bcopy{n}{2},2}(x,z)$ and $(q,(q',$ $q''))\in \Delta$
   \item two nodes $(y,q')$ and $(z,q'')$ such that 
$\psi_{\bcopy{n}{2},1}(z,x)$, 
$\psi_{\bcopy{n}{2},2}(z,y)$ and $((q,q'),$ $q'')\in \Delta$
  \end{itemize}
  
  Then, only nodes $(x,q)$ with $q$ final are childless and are those labelled 
with final states. As well, only $(x,q)$ with $q$ initial are fatherless.
  
  Then each $D_i$ is recognised by $A$ with this labelling.
 \qed
 \end{proof}
 
 \begin{lemma}
  $t$ is obtained by applying the $D_i$ to disjoint positions of $s$.
 \end{lemma}

 \begin{proof}
 We show by induction that $t' = D_{(j)} (s)$ if and only if $X_{t'} = 
X_s \cup \{x \mid (x,q) \in O_D\} \backslash \{x \mid (x,q) \in 
I_D\}$:

\begin{itemize}
 \item If $D = \emptydag$, it is true, as $X_{t'} = X_s$ and $t' = s$.
 \item If $D = (F \cdot_{1,1} D_\theta) \cdot_{1,1} G$, by 
induction hypothesis, we consider $r$ such that $r = F_{(j)}(s)$, we then 
have $X_r = X_s \cup \{y\} \backslash \{x \mid (x,q) \in I_F\}$, where 
$(y,q')$ is the only output node of $F$. By construction, the input node of 
$G$, $(z,q'')$ is such that $\psi_\theta(y,z)$, and thus we have $r' = 
\theta_{(j)}(r)$ such that $X_{r'} = X_r \backslash \{y\} \cup \{z\}$. By 
induction hypothesis, we have $X_{t'} = X_{r'} \cup \{x \mid (x,q) \in 
O_G\} \backslash \{z\}$, as $t' = G_{(j)} (\theta_{(j)} 
(F_{(j)}(s))) = G_{(j)}(r')$. Thus, $X_{t'} = X_s \cup \{x \mid (x,q) 
\in O_G\} \backslash \{x \mid (x,q) \in I_F\} = X_s \cup \{x \mid 
(x,q) 
\in O_D\} \backslash \{x \mid (x,q) \in I_D\}$.
\end{itemize}
  
  The other cases are similar and are thus left to the reader.
  It then suffices to construct this way successively $t_1 = {D_1}_{(i_1)} 
(s)$, $t_2 = {D_2}_{(i_2)}(t_1)$, etc, to obtain $t$ and prove the lemma.
   \qed
 \end{proof}

 We have proved both directions: for every $n$-stack trees $s$ and
 $t$, there exists a set of operations $D_i$ recognised by $A$ such
 that $t$ is obtained by applying the $D_i$ to disjoint positions
 of $s$ if and only if $\phi(X_s,X_t)$.
 \qed
\end{proof}

We then have a monadic interpretation with finite sets (all sets are
finite), and then, the graph has a decidable FO theory, which
concludes the proof.

\section{Example of a language}

We can see a rewriting graph as a language acceptor in a classical way by 
defining some initial and final states and labelling the edges.
%
%
%
We present here an example of a language recognised by a stack tree rewriting 
system. The recognised language is $\{u \shuffle u \mid u \in \Sigma\}$. 
%
Fix an alphabet $\Sigma$ and two special symbols $\uparrow$ and 
$\downarrow$. We consider $\TS_2(\Sigma\cup\{\uparrow,\downarrow\})$.
We now define a rewriting system $R$, whose rules are given in Fig.
\ref{fig:rules}.  

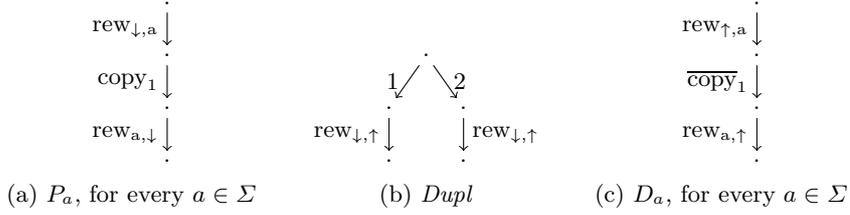
\begin{figure}[ht]
 \begin{center}
  \subfloat[$P_a$, for every $a\in \Sigma$]{
    \begin{minipage}[t]{0.3\linewidth}\centering
      \begin{tikzpicture}[scale=.7]
        \node (1) at (0,.5) {.}; 
        \node (2) at (0,-.5) {.}; 
        \node (3) at (0,-1.5) {.}; 
        \node (4) at (0,-2.5) {.}; 
        \draw [->] (1) to node[midway,left]{$\rew{\downarrow}{a}$} (2); 
        \draw [->] (2) to node[midway,left]{$\cop{1}$} (3); 
        \draw [->] (3) to node[midway,left]{$\rew{a}{\downarrow}$} (4);
      \end{tikzpicture}
    \end{minipage}
  }
  \subfloat[$\mathit{Dupl}$]{
    \begin{minipage}[t]{0.3\linewidth}\centering
      \begin{tikzpicture}[scale=.7]
	 \node (2) at (0,1) {.};
         \node (3) at (-.7,0) {.};
         \node (4) at (.7,0) {.};
         \node (5) at (-.7,-1) {.};
         \node (6) at (.7,-1) {.};
        \draw [->] (2) to node[midway,left]{$1$} (3);
        \draw [->] (2) to node[midway,right]{$2$} (4);
        \draw [->] (3) to node[midway,left] {$\rew{\downarrow}{\uparrow}$} (5);
        \draw [->] (4) to node[midway,right] {$\rew{\downarrow}{\uparrow}$} (6);
      \end{tikzpicture}
    \end{minipage}
  }
  \subfloat[$D_a$, for every $a\in \Sigma$]{
    \begin{minipage}[t]{0.3\linewidth}\centering
      \begin{tikzpicture}[scale=.7]
        \node (1) at (0,.5) {.}; 
        \node (2) at (0,-.5) {.}; 
        \node (3) at (0,-1.5) {.}; 
        \node (4) at (0,-2.5) {.}; 
        \draw [->] (1) to node[midway,left]{$\rew{\uparrow}{a}$} (2); 
        \draw [->] (2) to node[midway,left]{$\ncop{1}$} (3); 
        \draw [->] (3) to node[midway,left]{$\rew{a}{\uparrow}$} (4);
      \end{tikzpicture}
    \end{minipage}
  }
 \end{center}
 \caption{The rules of the rewriting system}
 \label{fig:rules}
\end{figure}

To recognise a language with this system, we have to fix an initial set of 
stack trees and a final set of stack trees. We will have a unique initial tree 
and a recognisable set of final trees. They are depicted on Fig.
\ref{fig:init-fin}.

\begin{figure}[ht]
  \begin{center}
  \subfloat[The initial tree.]{
  \begin{minipage}[t]{0.3\linewidth}\centering
    \begin{tikzpicture}[scale=.7]
      \node (root) at (0,0) {$\stack{1}{\downarrow}$}; 
    \end{tikzpicture}
  \end{minipage}
  }
  \subfloat[A final tree. $s$ is an arbitrary $1$-stack]{
  \begin{minipage}[t]{0.3\linewidth}\centering
    \begin{tikzpicture}[scale=.7]
      \node (root) at (0,0) {$s$}; 
      \node (0) at (-1.5,-1.5) {$\stack{1}{\uparrow}$}; 
      \node (1) at (1.5,-1.5) {$\stack{1}{\uparrow}$};
      \draw [->] (root) to (0); 
      \draw [->] (root) to (1); 
    \end{tikzpicture}
  \end{minipage}
  }
  \end{center}
  \caption{The initial and final trees.}
  \label{fig:init-fin}
\end{figure}
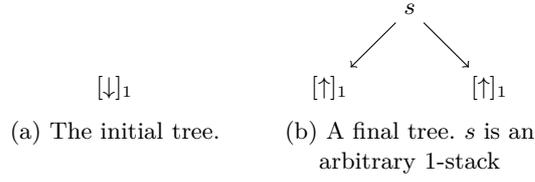

A word $w\in R^*$ is accepted by this rewriting system if there is a path from 
the initial tree to a final tree labelled by $w$.
The trace language recognised is 
\[
\{P_{a_1}\cdots P_{a_n} \cdot \mathit{Dupl} 
\cdot ((D_{a_n} \cdots D_{a_1}) \shuffle (D_{a_n} \cdots D_{a_1})) \mid 
a_1,\cdots ,a_n \in \Sigma\}.
\]
Let us informally explain why. We start on the initial tree, which has only a 
leaf labelled by a stack whose topmost symbol is $\downarrow$. So we cannot 
apply a $D_a$ to it. If we apply a $P_a$ to it, we remain in the same 
situation, but we added an $a$ to the stack labelling the unique node. So we 
can 
read a sequence $P_{a_1}\cdots P_{a_n}$. 
From this situation, we can also apply a $\mathit{Dupl}$, which yields a tree 
with three nodes whose two leaves are labelled by $\stack{1}{a_1\cdots a_n 
\uparrow}$, if we first read $P_{a_1}\cdots P_{a_n}$.
From this new situation, we can only apply $D_a$ rules. If the two leaves are 
labelled by $\stack{1}{b_1\cdots b_m \uparrow}$ and $\stack{1}{c_1\cdots 
c_\ell \uparrow}$, we can apply $D_{b_m}$ or $D_{c_\ell}$, yielding the same 
tree in which we removed $b_m$ or $c_\ell$ from the adequate leaf. We can do 
this until a final tree remains.
So, on each leaf, we will read $D_{a_n} \cdots D_{a_1}$ in this order, but we 
have no constraint on the order we will read these two sequences. So we 
effectively can read any word in $(D_{a_n} \cdots D_{a_1}) \shuffle (D_{a_n} 
\cdots D_{a_1})$. And this is the only way to reach a final tree.

To obtain the language we announced at the start, we just have to define a 
labelling $\lambda$ of each operation of $R$ as follows:
$\lambda(\mathit{Dupl}) = \varepsilon$, for every $a\in \Sigma$,
$\lambda(P_a) = \varepsilon$ and $\lambda(D_a) = a$, and remark that 
if $w$ is of the previous form, then $\lambda(w) = (a_1 \cdots a_n) \shuffle 
(a_1 \cdots a_n)$, and we indeed recognise $\{u\shuffle u \mid u \in 
\Sigma\}$.
